\documentclass[11pt]{article}

\def\comments{1}

\usepackage{preamble}

\title{The Limits of Pan Privacy and Shuffle Privacy \\ for Learning and Estimation}
\author{
    Albert Cheu\thanks{Khoury College of Computer and Information Sciences, Northeastern University.  Supported by NSF grants CCF-1750640, CNS-1816028, CNS-1916020.  \url{cheu.a@northeastern.edu}.}
    \and Jonathan Ullman\thanks{Khoury College of Computer and Information Sciences, Northeastern University.  Supported by NSF grants CCF-1750640, CNS-1816028, CNS-1916020.  \url{jullman@ccs.neu.edu}.}
}
\date{}

\begin{document}

\maketitle

\begin{abstract}
    There has been a recent wave of interest in intermediate trust models for differential privacy that eliminate the need for a fully trusted central data collector, but overcome the limitations of local differential privacy.  This interest has led to the introduction of the shuffle model (Cheu et al., EUROCRYPT 2019; Erlingsson et al., SODA 2019) and revisiting the pan-private model (Dwork et al., ITCS 2010).  The message of this line of work is that, for a variety of low-dimensional problems---such as counts, means, and histograms---these intermediate models offer nearly as much power as central differential privacy.  However, there has been considerably less success using these models for high-dimensional learning and estimation problems.  
    
    In this work we prove the first non-trivial lower bounds for high-dimensional learning and estimation in both the pan-private model and the general multi-message shuffle model.  Our lower bounds apply to a variety of problems---for example, we show that, private agnostic learning of parity functions over $d$ bits requires $\Omega(2^{d/2})$ samples in these models, and privately selecting the most common attribute from a set of $d$ choices requires $\Omega(d^{1/2})$ samples, both of which are exponential separations from the central model.  Our work gives the first non-trivial lower bounds for learning and optimization in both the pan-private and the general multi-message shuffle model.
\end{abstract}

\section{Introduction} \label{sec:intro}
The most widely accepted way to ensure individual privacy in the context of statistics and machine learning is \emph{differential privacy}~\cite{DworkMNS06}, which provides a strong guarantee that no individual user's data has a strong influence on the \emph{output} of the computation that are visible to the attacker.  Differentially private algorithms, however, are designed for a variety of different \emph{trust models} that determine what output is visible.  The strongest, and most commonly studied trust model is the \emph{central model}, in which a single party is entrusted to collect raw data from the users, runs a differentially private computation, and only the final output of this computation is visible.  On the other extreme, the weakest trust model is the \emph{local model}~\cite{KasiviswanathanLNRS08}, where we don't trust anyone to safeguard raw data, so each user applies differential privacy locally to their own data to compute a response, and each user's response is visible.  While the central model allows for many powerful algorithms, the local model is much less powerful (\cite{KasiviswanathanLNRS08, BeimelNO08, ChanSS12, DuchiJW13} et seq.) and significantly limits the accuracy of computations.

In principle there is no tradeoff between trust and power, as the user's can use cryptographic secure multiparty computation to implement any algorithm designed for the central model without any trusted party.  However, general-purpose secure multiparty computation has several drawbacks, such as large computation and communication costs, multiple rounds of interaction, and requiring all users to remain live throughout the computation.  Although there are more practical protocols implementing certain differentially private algorithms (\cite{DworkKMMN06} et seq.) so far these are restricted to relatively simple computations and are not practical for large-scale applications.

Thus, a recent focus has been on \emph{intermediate trust models} that offer some of the best features of both the central model and the local model.  Two models that have received significant attention are:
\begin{itemize}
    \item The \emph{shuffle model}~\cite{CheuSUZZ19,ErlingssonFMRTT19}.\footnote{More precisely, we consider a version of the shuffle model with an additional \emph{robustness} property~\cite{BalcerCJM20}.  Although the property is not without loss of generality, and has not always formalized in the literature, it is satisfied by all known natural shuffle protocols, and was one of the explicit motivations of studying the shuffle model~\cite{CheuSUZZ19}.  For brevity we use only the term ``shuffle model'' in the introduction, and defer more discussion of this issue to Section~\ref{sec:prelims}.}  In this model, users introduce randomness into their own data, as in the local model.  However the user's responses are then passed through a \emph{secure shuffler} so the responses are visible but not identified with individual users.  We consider the most general \emph{multi-message} shuffle model where each user can send multiple responses that are shuffled independently. An equivalent model would use secure aggregation to ensure that only a histogram of the responses is visible.  Secure shuffling and secure aggregation are significantly easier to achieve than general secure computation, and Google's \textsc{prochlo} system~\cite{Bittau+17} is a scalable realization of this model.
    
    \item The \emph{pan-private model}~\cite{DworkNPRY10}.  In this model, the users' data is processed in an online fashion by a central party.  We trust this central party to process the data but not to store it in perpetuity, so we assume that at any one point in the stream, the party's internal state may become visible.  This model captures, for example, a data collector who is well intentioned, and can be trusted to see raw data during process, but whose storage may be subject to breaches~\cite{AminJM20}.
\end{itemize}

We visualize the models in Figure~\ref{fig:models}. At first glance, these two models seem unrelated, however a recent result of Balcer, Cheu, Joseph, and Mao~\cite{BalcerCJM20} shows that, for a large class of problems that includes all the problems we study, any protocol in the shuffle model can be simulated in the pan-private model with only a small reduction in accuracy.  So for purposes of this work, we can think of these models as being ordered from least powerful to most powerful as
$
\textit{local} \preceq \textit{shuffle} \preceq \textit{pan-private} \preceq \textit{central}.
$

\begin{figure}[t!]
    \centering \hspace{-10pt}
    \includegraphics[width=0.55\textwidth]{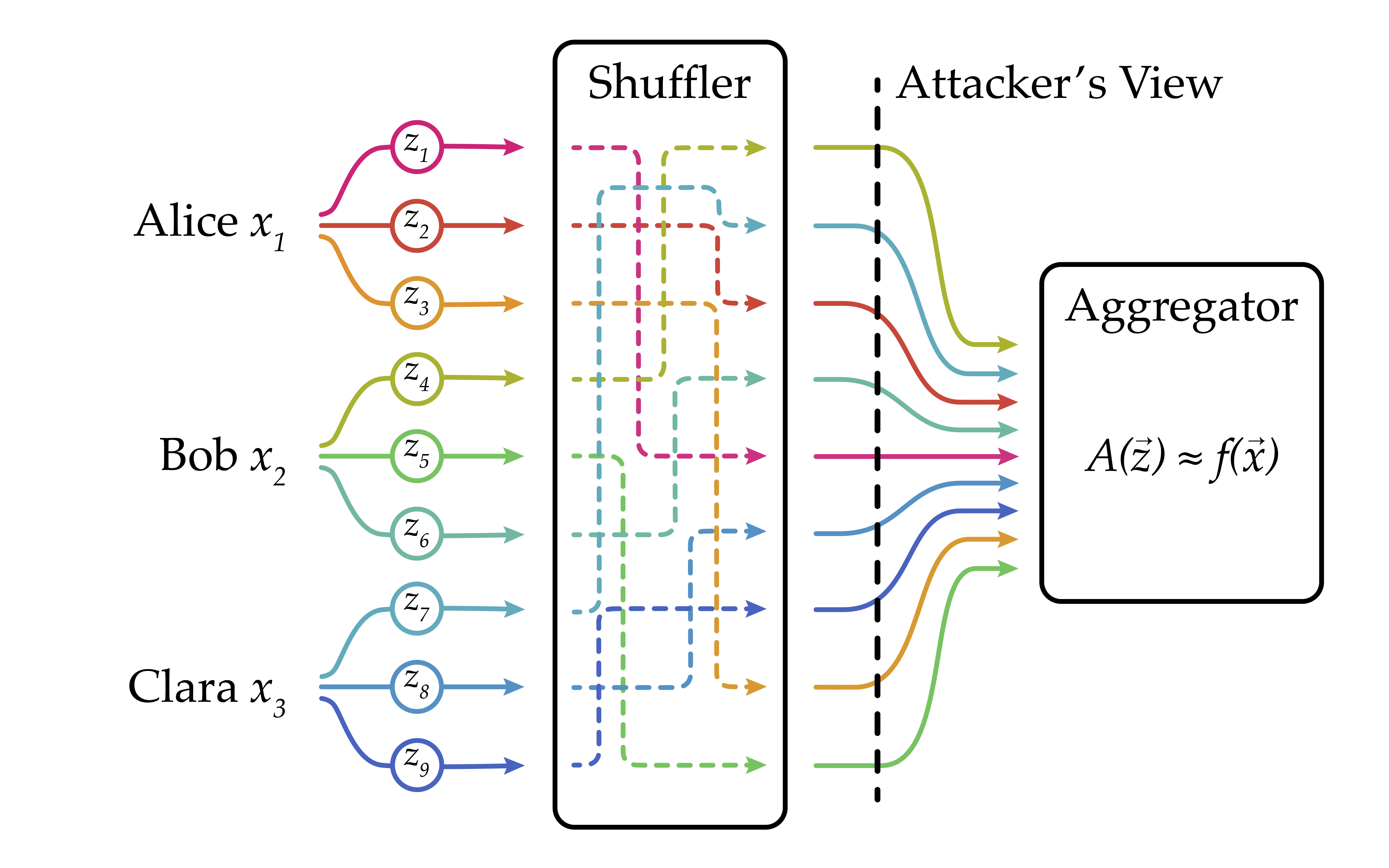} \hspace{30pt} 
    \includegraphics[width=0.35\textwidth]{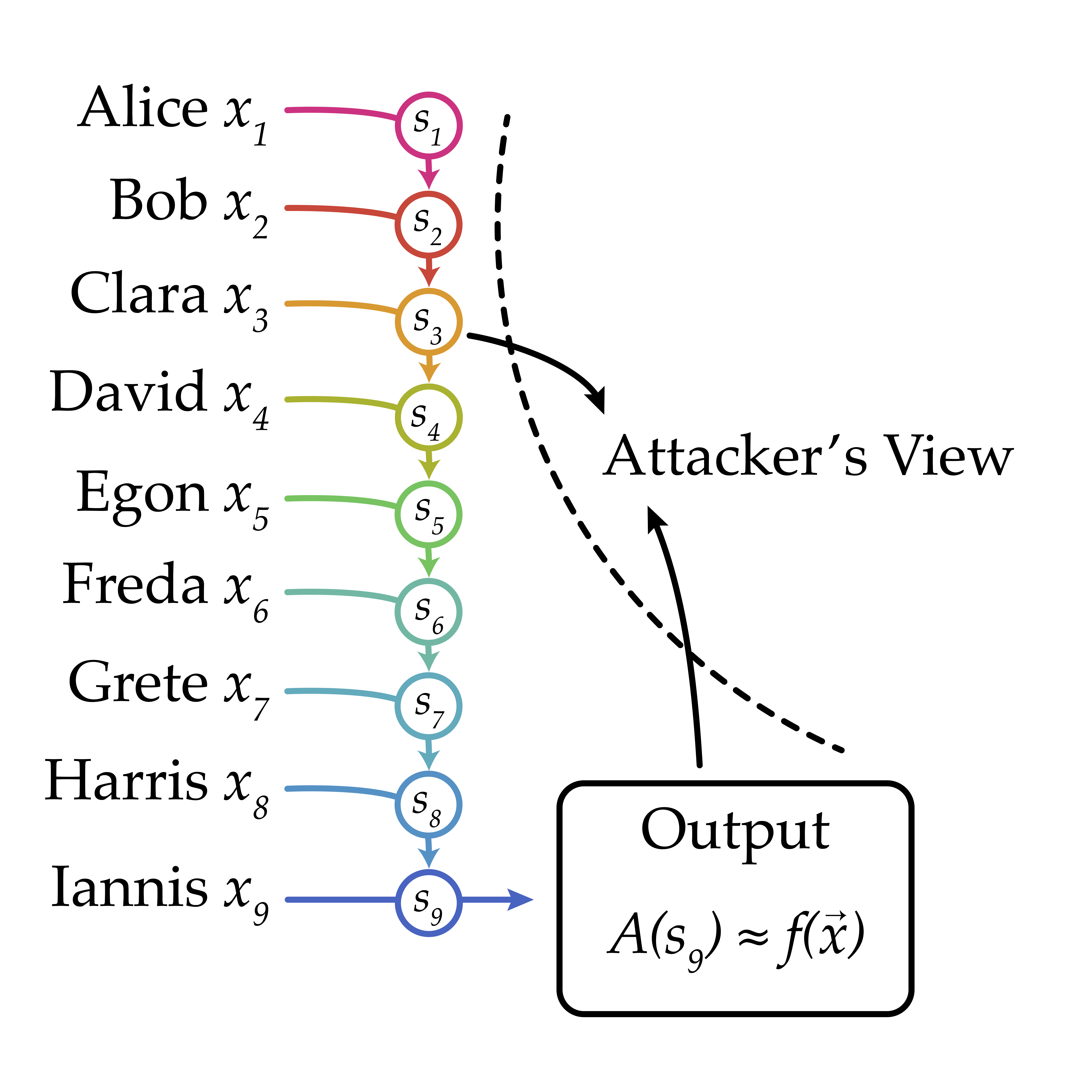}
    \caption{(Left) The multi-message shuffle model. The attacker's view consists of the entire set of messages, randomly shuffled.  (Right) The pan-privacy model.  The attacker's view consists of the output and the internal state at any single step, which is $s_3$ in this example.}
    \label{fig:models}
\end{figure}

Both the shuffle model (\cite{CheuSUZZ19, ErlingssonFMRTT19} et seq.) and the pan-private model~\cite{DworkNPRY10, MirMNW11, AminJM20} provably allow much greater accuracy than the local model, while also requiring weaker trust than the central model.  See Section~\ref{sec:rw} for a more specific overview of recent progress.  However, these positive results are mostly limited to relatively simple functionalities, such as computing means and histograms over the user's data.  We note that these are all problems that can be solved efficiently in the local model with reasonable, although larger, sample complexity.  However, for problems such as learning parities and selecting the most common attribute, where the local model where the local model is most severely limited~\cite{KasiviswanathanLNRS08,DuchiJW13,Ullman18,EdmondsNU20}, there is no evidence that either the pan-private or shuffle model can overcome these limitations. Our main contribution is to show  that these limitations are inherent:
\begin{quotation}
    \noindent\emph{For many high-dimensional learning and estimation problems, the shuffle and pan-private models incur an exponential cost in sample complexity relative to the central model.}
\end{quotation}
For those familiar with differential privacy, our results can be interpreted as the statement \emph{there is no analogue of the exponential mechanism in the pan-private or shuffle models}, as we prove lower bounds for problems that can be solved in the central model by applying the exponential mechanism.

Our specific lower bounds follow from a new general lower bound argument.  We note that the two most common lower bounds techniques for the local model cannot prove lower bounds for the pan-private and shuffle models, so our lower bounds cannot be proven by any straightforward extension of existing lower bound techniques.  Specifically, there is no non-trivial upper bound on the mutual information between the algorithm's inputs and outputs~\cite{BalcerC20}, so information-theoretic arguments~\cite{McGregorMPRTV10,DuchiJW13} do not apply.  Moreover, these models can solve problems that would requite infinitely many statistical queries to solve, so the simulation of the local model in the statistical query model~\cite{KasiviswanathanLNRS08} cannot be extended to these more general models.

\subsection{Results}

Our main results are lower bounds for many closely related learning and estimation problems in both the pan-privacy and shuffle models of differential privacy.  We note that throughout this work we adopt the standard model for studying privacy for distributional problems where we define the \emph{accuracy} goal with respect to input satisfying certain distributional assumptions, but define \emph{privacy} for a worst-case dataset.  We begin by highlighting two important cases of our results.

\mypar{Learning Parities.} In this canonical learning problem, we are given a dataset consisting of $n$ labeled examples $\{(x_i, y_i)\}$ sampled from some distribution $\bP$ over the domain $\pmo^d \times \pmo$.  The goal is to output a parity function $h_{S}(x) = \prod_{j \in S} x_j$ that predicts the labels nearly as well as any other parity function.  Namely, 
$$
\pr{(x,y) \sim \bP}{h_{S}(x) = y} \geq \max_{T} \pr{(x,y) \sim \bP}{h_{T}(x) = y} - \alpha.
$$
In the central model this problem can be solved privately to any constant level of accuracy with just $O(d)$ samples~\cite{KasiviswanathanLNRS08}, whereas in the local model any algorithm solving this problem requires $\Omega(2^d)$ samples~\cite{KasiviswanathanLNRS08,EdmondsNU20}.\footnote{For specificity, we state lower bounds for the \emph{non-interactive} local model of differential privacy, although, for every problem we consider, slightly weaker bounds are known to hold for interactive variants of the local model as well.}  We prove an exponential separation between the central model and the pan-privacy and shuffle models, showing that, for learning parities, these models are much more similar to the local model.
\begin{thm}(Informal)
	Any differentially private algorithm that leans parity functions to constant accuracy in the pan-privacy model or the shuffle privacy model requires $\Omega(2^{d/2})$ samples in the worst-case.
\end{thm}
We also consider learning \emph{sparse} parities, where our goal is to output some $k$-sparse parity function $h_{S}$, $|S|\leq k$ that competes with the best parity function on $k$ variables.  That is, 
$$
\pr{(x,y) \sim \bP}{h_{S}(x) = y} \geq \max_{T : |T| \leq k} \pr{(x,y) \sim \bP}{h_{T}(x) = y} - \alpha.
$$
We show that learning $k$-sparse parities requires $\Omega(\sqrt{\binom{d}{\leq k}})$ samples where $\binom{d}{\leq k}$ denotes the number of $k$-sparse parity functions on $d$ bits.

\mypar{Selection.} One of the most celebrated tools in central-model differential privacy is the \emph{exponential mechanism} of McSherry and Talwar~\cite{McSherryT07}, which is a very general and very accurate method for optimizing a Lipschitz loss function over a discrete set of choices.  The canonical problem solved by the exponential mechanism is the following \emph{selection} problem: given a dataset consisting of $n$ samples $\{x_i\}$ from some distribution $\bP$ over the domain $\{0,1\}^d$, select a coordinate $j$ such that the expected value of the $j$-th coordinate is as large as possible.  Namely,
$$
\ex{x \sim \bP}{x_j} \geq \max_{k} \ex{x \sim \bP}{x_k} - \alpha.
$$
In the central model, the exponential mechanism solves this problem to any constant level of accuracy with just $O(\log d)$ samples, whereas in the local model any algorithm solving this problem requires $\Omega(d \log d)$ samples~\cite{DuchiJW13, Ullman18}.  Again, we show an exponential separation between the central model and the pan-privacy and shuffle models, demonstrating that there is no general-purpose analogue of the exponential mechanism in these intermediate models.
\begin{thm}(Informal)
	Any differentially private algorithm that solves selection to constant accuracy in the pan-privacy model or the shuffle privacy model requires $\Omega(\sqrt{d})$ samples in the worst-case.
\end{thm}

\mypar{Variants of Differential Privacy.} We emphasize that all of our lower bounds hold for the most general variant of differential privacy, $(\eps,\delta)$-differential privacy for $\delta \leq 1/n^{1.1}$, and obtain lower bounds for this variant is one of the main technical challenges addressed by our work.  Thus, our results imply essentially the same lower bounds for pure differential privacy, concentrated differential privacy~\cite{DworkR16,BunS16}, truncated concentrated differential privacy~\cite{BunDRS18}, R\'{e}nyi differential privacy~\cite{Mironov17}, and Gaussian differential privacy~\cite{DongRS19}, none of which were known prior to our work.

\mypar{More Applications.} In our work we also prove tight lower bounds for several closely related, natural problems that have been studied in the literature on differential privacy:
\begin{itemize}
	\item \emph{Estimating $k$-Sparse Parities} for $1 \leq k \leq d$.  Here we are given samples $\{x_i\}$ from a distribution $\bP \in \pmo^d$, and the goal is to output a set of estimates $\{ a_{T} \}_{T \subseteq [d] \atop |T| \leq k}$ such that $$\left| a_{T} - \ex{x \sim \bP}{\prod_{j \in T} x_j} \right| \leq \alpha$$ for every $T$.
	
	\item \emph{$d$-wise Simple Hypothesis Testing}.  Here we are given samples $\{x_i\}$ from a distribution $\bP \in \cQ$ where $\cQ = \{\bQ_1,\dots,\bQ_d\}$ is a known set of $d$ hypotheses satisfying $\dtv(\bQ_{i},\bQ_{j}) \geq \alpha$, and the goal is to determine which of these distributions is $\bP$.
	
	\item \emph{1-Sparse Mean Estimation}. Here we are given samples $\{x_i\}$ from a distribution $\bP \in \pmo^d$ with mean $\mu$, with the promise that $\| \mu \|_0 = 1$, and the goal is to output $\hat\mu$ such that $\| \mu - \hat\mu \|_\infty \leq \alpha$.
\end{itemize}
We summarize our lower bounds and compare to the local and central models in Table~\ref{tab:main}.  We also stress that, while the focus of this work is on lower bounds and not algorithms, all of our lower bounds are easily seen to be tight up to logarithmic factors with respect to trivial statistical query algorithms~\cite{Kearns93} that can be implemented in both the pan-private and shuffle models of privacy.

\begin{table}[t!]
\centering
\begin{tabular}{|c|c|c|c|c|}
\hline
\textbf{Problem} &
  \textbf{Parameters} &
  \textbf{Local Privacy} &
  \textbf{\begin{tabular}[c]{@{}c@{}}Pan/Shuffle Privacy\\ (This Work)\end{tabular}} &
  \textbf{Central Privacy} \\ \hline
  &&&& \\[-11pt]
\begin{tabular}[c]{@{}c@{}}Learning \\ Parities\end{tabular} &
  \begin{tabular}[c]{@{}c@{}}Dimension $d$\\Sparsity $k$\end{tabular} &
  \begin{tabular}[c]{@{}c@{}}$\Omega(\binom{d}{\leq k} \log \binom{d}{\leq k})$\\ \cite{EdmondsNU20}\end{tabular} &
  \begin{tabular}[c]{@{}c@{}}$\Omega\left(\sqrt{\binom{d}{\leq k}}\right)$\\ Thms \ref{thm:learning-parities-pan}/\ref{thm:learning-parities-shuffle}\end{tabular} &
  \begin{tabular}[c]{@{}c@{}}$O(\log \binom{d}{\leq k})$\\ \cite{KasiviswanathanLNRS08}\end{tabular} \\ \hline
    &&&& \\[-8pt]
Selection &
  Dimension $d$ &
  \begin{tabular}[c]{@{}c@{}}$\Omega(d \log d)$\\ \cite{DuchiJW13}\end{tabular} &
  \begin{tabular}[c]{@{}c@{}}$\Omega(\sqrt{d})$\\ Thms \ref{thm:selection-pan}/ \ref{thm:selection-shuffle}\end{tabular} &
  \begin{tabular}[c]{@{}c@{}}$O(\log d)$\\ \cite{McSherryT07}\end{tabular} \\[12pt] 
  \hline
  &&&& \\[-11pt]
\begin{tabular}[c]{@{}c@{}}Estimating \\ Parities\end{tabular} &
  \begin{tabular}[c]{@{}c@{}}Dimension $d$\\ Sparsity $k$\end{tabular} &
  \begin{tabular}[c]{@{}c@{}}$\Omega(\binom{d}{\leq k} \log \binom{d}{\leq k})$\\ \cite{EdmondsNU20}\end{tabular} &
  \begin{tabular}[c]{@{}c@{}}$\Omega\left(\sqrt{\binom{d}{\leq k}}\right)$\\ Thms \ref{thm:estimate-parities-pan}/\ref{thm:estimate-parities-shuffle}\end{tabular} &
  \begin{tabular}[c]{@{}c@{}}$\tilde{O}(\sqrt{d} \log \binom{d}{\leq k})$\\ \cite{HardtR10}\end{tabular} \\ \hline
  &&&& \\[-8pt]
\begin{tabular}[c]{@{}c@{}}$d$-Wise Simple \\ Hypothesis Testing\end{tabular} &
  $d$ Hypotheses &
  \begin{tabular}[c]{@{}c@{}}$\Omega(d \log d)$\\ \cite{GopiKKNWZ20}\end{tabular} &
  \begin{tabular}[c]{@{}c@{}}$\Omega(\sqrt{d})$\\ Thms \ref{thm:hyp-selection-pan}/\ref{thm:hyp-selection-shuffle}\end{tabular} &
  \begin{tabular}[c]{@{}c@{}}$O(\log d)$\\ \cite{BunKSW19}\end{tabular} \\[12pt] \hline
  &&&& \\[-8pt]
\begin{tabular}[c]{@{}c@{}} 1-Sparse\\ Mean Estimation\end{tabular} &
  Dimension $d$ &
  \begin{tabular}[c]{@{}c@{}}$\Omega(d \log d)$\\ \cite{DuchiJW13}\end{tabular} &
  \begin{tabular}[c]{@{}c@{}}$\Omega(\sqrt{d})$\\ Thms \ref{thm:sparse-mean-pan}/\ref{thm:sparse-mean-shuffle}\end{tabular} &
  \begin{tabular}[c]{@{}c@{}}$O(\log d)$ \\ $[$Folklore$]$ \end{tabular} \\[12pt] \hline
\end{tabular}
\caption{
    Summary of our sample-complexity lower bounds for the pan-privacy and shuffle privacy models, in comparison to the local and central models.  For brevity, all lower bounds are stated for accuracy $\alpha = 1/100$ and $(1,n^{-2})$-differential privacy.  See the formal theorems for more general statements.  All our lower bounds are tight up to polylogarithmic factors.  We use the notation $\binom{d}{\leq k} = \sum_{i = 1}^{k} \binom{d}{i}$.  
}
\label{tab:main}
\end{table}

\subsection{Techniques}
Our results are all a consequence of a very general lower bound for algorithms in these models.  For simplicity, we will restrict this discussion to pan-private algorithms, as lower bounds for shuffle privacy will then follow from a general transformation from the shuffle model to the pan-privacy model due to Balcer, Cheu, Joseph, and Mao~\cite{BalcerCJM20}.  Also, in this discussion we will ignore the parameter $\delta$ for brevity, but, crucially, our results apply for moderately small $\delta > 0$.

Let $\{\bP_{v}\}_{v \in \cV}$ be some family of distributions over the domain $\cX$, let $V$ be uniform over $\cV$, and let
$$
\bU = \ex{v \sim V}{\bP_{v}}
$$
be the uniform mixture of these distributions.  We will give lower bounds that show no $(\eps,\delta)$-differentially private algorithm in the pan-private or shuffle models can distinguish $n$ i.i.d.\ samples from $\bU^n$ from data drawn from the mixture $\bP_V^n$, where we chose $v \sim V$ uniformly and then sample from $\bP_v^n$.  We will, of course, choose the family $\{\bP_{v}\}$ so that any algorithm solving one of the problems above, must distinguish $\bU^n$ from $\bP_V^n$, which is how we will obtain sample-complexity lower bounds.

For background, let's recap the way to use this setup to prove lower bounds in the (non-interactive) local model of differential privacy.  Here, one chooses the data from the mixture $\bP_{V}^n$, and a lemma of Duchi, Jordan, and Wainwright~\cite{DuchiJW13} gives a bound on the mutual information between the output of the protocol $\Pi$ and the identity of the random mixture component $V$:
\begin{equation} \label{eq:djw}
I(\Pi(\bP_V^n) ; V) = O(n \cdot \eps^2 \cdot \| \{ \bP_{v} \} \|_{\infty \to 2}^2)
\end{equation}
where
\begin{equation*}
\|\{\bP_v\}\|_{\infty \to 2}^2 = \sup_{f \from \cX \to [\pm 1]} \ex{v \sim V}{\left( \ex{x \sim \bP_v}{f(x)} - \ex{x \sim \bU}{f(x)}  \right)^2}
\end{equation*}
is the crucial quantity determining how hard these distributions are to distinguish subject to local differential privacy.  For intuition, note that this quantity satisfies the relationship
$$
\|\{\bP_v\}\|_{\infty \to 2}^2 \leq \ex{v \sim V}{\sup_{f : \cX \to [\pm 1]} \left( \ex{x \sim \bP_v}{f(x)} - \ex{x \sim \bU}{f(x)}  \right)^2} = 4 \cdot \ex{v \sim V}{ \dtv(\bP_v, \bU)^2 },
$$
but it can be much smaller than $4 \cdot \mathbb{E}_{v \sim V}(\dtv(\bP_v, \bU)^2)$, which is crucial for proving tight lower bounds.

Given this lemma, and a construction of a hard distribution family such that $\|\{\bP_v\}\|_{\infty \to 2}^2$ is small, it is not hard to deduce a lower bound on the number of samples $n$ required to identify the specific mixture component $V$.  It's also not too difficult to construct a family of hard distributions for all of our problems of interest (see Section~\ref{sec:hard-distributions}).  We note that all of the lower bounds in the ``local model'' column of Table \ref{tab:main} are proven via this approach.

With this state-of-affairs, it's tempting to try to argue that a mutual-information bound analogous to \eqref{eq:djw} holds for pan-private or shuffle model algorithms.  However, Balcer and Cheu~\cite{BalcerC20} constructed a family of distributions and a pan-private algorithm such that the mutual information $I(\Pi; V)$ can be unbounded, showing that the purely information-theoretic approach used to prove lower bounds for the local model cannot work for pan-privacy.\footnote{The algorithm showing pan-private algorithms can have unbounded mutual information crucially uses the full generality of $(\eps,\delta)$-differential privacy for $\delta > 0$, however, even for stricter variants of differential privacy where the mutual information is bounded, we don't know how to obtain a mutual-information bound as strong as~\eqref{eq:djw} for any of these variants.}

Nonetheless, we prove the following indistinguishability lemma for pan-private algorithms:
\begin{equation} \label{eq:main-lem}
\dtv(\Pi(\bU^n), \Pi(\bP_{V}^n)) \leq O(n \cdot \eps \cdot \|\{\bP_v\}\|_{\infty \to 2})
\end{equation}
Although this bound is quantitatively somewhat weaker than \eqref{eq:djw}---in ways that are actually crucial to avoid proving false statements---it is nonetheless sufficient to give tight lower bounds for all of the problems we consider.  The value of this lemma is that, even though the information-theoretic bounds that are used in the local model are false for the pan-private model, the exact same constructions of hard distributions can be used to obtain lower bounds for pan-privacy!  

The proof of this lemma uses a hybrid argument, where we transition between data sampled from $\bU^n$ and data sampled from $\bP_V^n$.  Namely, we fix a value of $i$ between $0$ and $n$ and consider the case where the first $i$ inputs are sampled from $\bU^i$ and the remaining $n-i$ inputs are sampled from $\bP^{n-i}$.  We then bound the total variation distance between the $i$-th case and the $(i+1)$-st case and apply the triangle inequality.  In each step, we carefully argue that the total variation distance between the two cases follows from a careful application of \eqref{eq:djw} to the algorithm that computes the internal state after viewing the first $i$ inputs, which is why we ultimately get a bound of a similar form.

\subsection{Related Work} \label{sec:rw}

\mypar{Comparison to the Concurrent Works of~\cite{ChenGKM20} and ~\cite{BeimelHNS20}.}
A concurrent and independent work of Chen, Ghazi, Kumar, and Manurangsi~\cite{ChenGKM20} proves lower bounds for selection and learning parity in the multi-message shuffle model.  Their lower bounds depend on the number of messages, and are only non-trivial when the number of messages is relatively small, whereas our lower bounds do not require any bound on the number of messages.  For example, their lower bound for selection is $\Omega(d/m)$, where $m$ is the number of messages, while our lower bound for selection is $\Omega(\sqrt{d})$ for any number of messages, and our lower bound is matched by a trivial algorithm that sends $d$ messages.  Compared to ours, their lower bounds do not require the shuffle protocol to be robust, although robustness was a motivating feature of the shuffle model that is discussed in the early work on the subject~\cite{CheuSUZZ19, ErlingssonFMRTT19}.  Their work also does not consider the pan-privacy model, and their arguments do not seem to apply to that model.

Another concurrent and independent work of Beimel, Haitner, Nissim, and Stemmer~\cite{BeimelHNS20} also proves lower bounds for multi-message shuffle protocols that use a small number of messages.  They show that if an $m$-message shuffle protocol is private when run with for $n$ users, then each user's messages reveals at most $\approx n^m$ bits of information about their input, which allows them to prove non-trivial lower bounds when $m$ is quite small.

\mypar{The Shuffle Model.} The shuffle model was introduced concurrently in works by Cheu et al.~\cite{CheuSUZZ19} and Erlingsson et al.~\cite{ErlingssonFMRTT19}.  These works were both inspired by Google's \textsc{prochlo} system~\cite{Bittau+17}, which implements a more general algorithmic paradigm called \emph{encode, shuffle, and analyze}.  Much of the work in this model has focused on constructing optimal algorithms for problems like binary sums~\cite{CheuSUZZ19, GhaziGKMPV20}, real-valued sums~\cite{BalleBGN19,GhaziPV19,GhaziMPV20,GhaziKMP20,BalleBGN20}, histograms and heavy-hitters~\cite{CheuSUZZ19, BalcerC20, GhaziGKPV20}, and uniformity testing~\cite{BalcerCJM20}.  Another complementary set of works have given general \emph{amplification theorems} showing that if each user applies a differentially private randomizer to their data, then the shuffle protocol using the randomizer satisfies differential privacy with stronger parameters \cite{BalleBGN19,ErlingssonFMRTT19}.

Almost all prior lower bounds for the shuffle model apply only to a special case of the model where each user sends only a single response, the so-called \emph{single-message shuffle model}.  Cheu et al.~\cite{CheuSUZZ19} showed that if a protocol is private in this restricted model, then each user's response satisfies local differential privacy, for which we already have strong lower bounds.  Their approach was refined by Ghazi et al.~\cite{GhaziGKPV20}, who obtained stronger bounds for single-message protocols.  Balle et al.~\cite{BalleBGN19} proved a lower bound for computing real-valued sums in the single-message model.  In contrast, our lower bounds hold for the general \emph{multi-message shuffle model}, where each user may send an arbitrary number of messages that are shuffled independently.  Note that in this model, the user's individual responses need not satisfy any local differential privacy~\cite{BalcerC20}.  An early lower bound for the multi-message shuffle model is due to Ghazi et al.~\cite{GhaziGKMPV20}, and applies to computing binary sums subject to pure differential privacy and a strong communication constraint.  We emphasize that our lower bounds do not impose any restriction on the number of messages or the amount of communication.

\mypar{The Pan-Private Model.} The pan-privacy model was introduced by Dwork et al.~\cite{DworkNPRY10} as a model of differential privacy for streaming algorithms, and they constructed pan-private algorithms for classic streaming problems like distinct elements.  Their algorithm was subsequently improved by Mir et al.~\cite{MirMNW11}, who also gave the first lower bounds for this model.  We note that their technique gives lower bounds for worst-case inputs, whereas our technique gives lower bounds for distributional problems.

More recently, Amin, Joseph, and Mao~\cite{AminJM20} revisited the model from the perspective of finding an intermediate trust model between local and central privacy, which is the perspective we adopt in this work.  They also gave an algorithm for \emph{uniformity testing} and a matching lower bound for algorithms satisfying pure differential privacy, which is $(\eps,\delta)$-privacy with $\delta = 0$.  Theirs is the first lower bound in this model for any distributional problem.  As we discussed above, their information-theoretic arguments are inherently limited to pure differential privacy, whereas ours apply to differential privacy in general.

The initial work on pan-privacy considered a more general model where the attacker can view the internal state at two or more arbitrary steps, however~\cite{AminJM20} showed that this model is equivalent to the local model with sequential interaction.  Our lower bounds apply to the weakest model, where the attacker can view the state at just a single time step.

\mypar{Lower Bounds Techniques in the Local and Central Model.} We briefly summarize the techniques for proving lower bounds in the more well studied models of differential privacy.  The first lower bounds for local differential privacy were proven by Kasiviswanathan et al.~\cite{KasiviswanathanLNRS08}, who proved that the local model is equivalent, up to polynomial factors, to the statistical queries model~\cite{Kearns93}.  Balcer and Cheu~\cite{BalcerC20} showed that the shuffle and pan-private model do not admit such a characterization.  Recently Edmonds, Nikolov, and Ullman~\cite{EdmondsNU20} gave a nearly tight characterization of the sample complexity of query release and agnostic learning in the non-interactive local model.  Subsequent work gave stronger lower bounds for specific problems in the local model~\cite{BeimelNO08, ChanSS12, DuchiJW13, BassilyS15, JosephKMW18, DuchiR18, DuchiR19, JosephMNR19}, including interactive variants of the local model.  This line of work primarily uses information-theoretic arguments that were first introduced by McGregor et al.~\cite{McGregorMPRTV10} in the context of two-party differential privacy.  However, these approaches cannot give strong lower bounds for the pan-private and shuffle model~\cite{BalcerC20}, and the main novelty in our work is finding strong lower-bound arguments for these intermediate models that do not require strong information bounds.

There are two main approaches to proving lower bounds for high-dimensional problems in the central model of differential privacy.  The first are reconstruction attacks, introduced by Dinur and Nissim (\cite{DinurN03} et seq.).  These attacks only apply when computing some statistics to very high accuracy, and thus cannot give non-trivial lower bounds for distributional problems where the accuracy can never be smaller than the sampling error.  The other main approach is based on tracing attacks (\cite{BunUV14, DworkSSUV15, SteinkeU17} et seq.).  Although tracing attacks give tight lower bounds for the central model, but the lower bounds we prove for more restricted models are exponentially larger, and do not seem to be provable using tracing attacks.  We refer the reader to~\cite{DworkSSU17-arsia} for a survey of these attacks lower bounds.
\section{Preliminaries} \label{sec:prelims}

\subsection{Notational Conventions}
We use boldface letters denote probability distributions, capital letters in plain math text denote random variables, and calligraphic letters denote sets. We reserve $M$ for randomized algorithms and $\Pi$ for distributed protocols.  Throughout this work, we use the notation $[k] := \{1, 2, \ldots, k\}$. 

\subsection{Differential Privacy}

We define a \emph{dataset} $\vec{x} \in \cX^n$ to be an ordered tuple of $n$ rows where each row is drawn from a data universe $\cX$ and corresponds to the data of one user. Two datasets $\vec{x},\vec{x}\,' \in \cX^n$ are \emph{neighbors}, denoted as $\vec{x} \sim \vec{x}\,'$, if they differ in at most one row.

\begin{defn}[Differential Privacy \cite{DworkMNS06}]
An algorithm $M: \cX^n \rightarrow \cR$ satisfies \emph{$(\eps, \delta)$-differential privacy} if, for every pair of neighboring datasets $\vec{x}$ and $\vec{x}\,'$ and every event $\cC \subseteq \cR$,
	$$\pr{}{M(\vec{x}\vphantom{'}) \in \cC} \le e^\eps \cdot \pr{}{M(\vec{x}\,') \in \cC} + \delta.$$
\end{defn}

The \emph{central model} of differential privacy refers to the case where the algorithm $M$ is allowed to depend arbitrarily on $\vec{x}$ with no further restrictions.
\subsection{The Pan-Private Model}

A \emph{pan-private} algorithm observes the data as a \emph{stream}.  At each step, the algorithm receives a datapoint that it uses to update its internal state, and this process repeats until the stream is exhausted and a final output is computed.  We say that two streams $\vec{x}$ and $\vec{x}\,'$ are \emph{neighbors} if they differ in at most one element.  Pan-privacy models an attacker who observes the final output of the algorithm, as well as the internal state at any one step in the stream, and requires that the joint distribution of these two pieces of information is differentially private.

\begin{defn}[Online Algorithm]
\label{def:online}
An \emph{online algorithm} $M$ is defined by a sequence of internal algorithms $M_1,M_2,\dots$ and an output algorithm $M_{\Ou}$. On input $\vec{x}$, the first function $M_1 \from \cX \to \In$ maps $x_1$ to a state $s_1$ and the remaining functions $M_{i}$ map $x_{i}$ and the previous state $s_{i-1}$ to a new state $s_i$. At the end of the stream, $M$ publishes a final output by executing $M_{\Ou}\from \In \to \Ou$ on its final internal state.
\end{defn}

\begin{defn}[Pan-privacy \cite{DworkNPRY10,AminJM20}]
\label{def:pan}
	Given an online algorithm $M$, let $M_{\In}(\vec{x})$ denote its internal state after processing stream $\vec{x}$, and let $\vec{x}_{\leq t}$ be the first $t$ elements of $\vec{x}$. We say $M$ is \emph{$(\eps,\delta)$-pan-private} if, for every pair of neighboring streams $\vec{x}$ and $\vec{x}\,'$, every time $t$ and every set of internal state, output state pairs $T \subseteq \In \times \Ou$,
    \begin{equation}
    \label{eq:pan}
    \pr{M}{\big(M_{\In}(\vec{x}_{\leq t}), M_{\Ou}(M_{\In}(\vec{x}))\big) \in T} \leq e^\eps\cdot\pr{M}{\big(M_{\In}(\vec{x}\,'_{\!\leq t}), M_{\Ou}(M_{\In}(\vec{x}\,'))\big) \in T} + \delta.
    \end{equation}
    See Figure \ref{fig:models} for a diagram.
\end{defn}

Note that any pan-private algorithm can trivially be implemented in the central model.
Our definition of pan-privacy is the specific variant given by Amin et al. \cite{AminJM20}. This version guarantees record-level privacy (uncertainty about the presence of any single stream element) rather than user-level privacy (uncertainty about the presence of any one data universe element). We use this variant because for the problems we consider it is natural to model each user as contributing a single element of the stream.

Lastly, note that when we consider 

\subsection{The Shuffle Model}
In the \emph{shuffle model}, each user individually randomizes their own data to produce a series of messages.  Unlike the local model, where these messages would be identified with the user who produced them, we allow the users to send their messages to a \emph{secure shuffler} that collects all the messages of all the users and randomly permutes them.\footnote{See \cite{Bittau+17} for a discussion of various choices of how to implement such a secure shuffler.}  The shuffle model captures an attacker who observes the messages after they are shuffled, and we require this shuffled set of messages to satisfy differential privacy.  An equivalent model would allow the attacker observes only a histogram of the messages.


\begin{defn}[Shuffle Model \cite{CheuSUZZ19}]
A protocol $\Pi$ in the \emph{shuffle model} consists of three randomized algorithms:
\begin{itemize}
\item
    A \emph{randomizer} $\Pi_R: \cX \rightarrow \cY^*$ mapping data to (possibly variable-length) vectors. The length of the vector is the number of messages sent. If, on all inputs, the probability of sending a single message is 1, then the protocol is said to be \emph{single-message}. Otherwise, the protocol is \emph{multi-message}.
\item
    A \emph{shuffler} $\Pi_S: \cY^* \to \cY^*$ that applies a uniformly random permutation to all messages. 
\item
    An \emph{analyzer} $\Pi_A: \cY^* \rightarrow \Ou$ that computes on a permutation of messages.
\end{itemize}
As the shuffler is the same in every protocol, we identify each shuffle protocol by $\Pi=(\Pi_R,\Pi_A)$. We define the honest execution on input $\vec{x}\in\cX^n$ as
$$
\Pi(\vec{x}) := \Pi_A(\Pi_S(\Pi_R(x_1), \dots, \Pi_R(x_n))).
$$
We denote the output of the shuffler as
$$
(\Pi_S\circ \Pi_R^n)(\vec{x}) := \Pi_S(\Pi_R(x_1), \dots, \Pi_R(x_n)).
$$
We assume that users and the analyzer have access to $n$, as well as an arbitrary amount of public randomness.
\end{defn}

It remains to define differential privacy in this model.  We note that the output of the shuffler only follows the distribution $\Pi(\vec{x})$ if all users are following the protocol as specified.  This assumption is undesirable because it means each user is reliant on other users to behave correctly.  Thus we consider a \emph{robust} variant of the shuffle model, where we require that the protocol remains private when only a constant fraction of users behave correctly, while the other users may behave arbitrarily.  We emphasize all known natural protocols in this model satisfy the additional robustness condition, and the need for robustness was explicitly discussed in~\cite{CheuSUZZ19} as a feature of the model, so we consider the robust variant to be the most appropriate version of the model.
\begin{defn} [Robust Shuffle Differential Privacy \cite{BalcerCJM20}]
\label{def:robust_shuffle_dp}
	Fix $\gamma\in(0,1]$. A protocol $\Pi=(R,A)$ is \emph{$(\eps,\delta, \gamma)$-robustly shuffle differentially private} if, for all $n\in\N$ and $\gamma' \geq \gamma$, the algorithm $\Pi_S \circ \Pi_R^{\gamma' n}$ is $(\eps, \delta)$-differentially private. In other words, $\Pi$ guarantees $(\eps, \delta)$-shuffle privacy whenever at least a $\gamma$ fraction of the intended number of users follow the protocol.
\end{defn}

We remark that the above definition only explicitly handles drop-out attacks, where malicious users send no messages. However, dropping out is the worst malicious users can do.  Combining arbitrary messages from malicious users with the messages of honest users can be viewed as a post-processing of $\Pi_S\circ \Pi_R^{\gamma n}$. If $\Pi_S\circ \Pi_R^{\gamma n}$ is already differentially private for the outputs of the $\gamma n$ users alone, then differential privacy's resilience to post-processing ensures that adding other messages does not affect this guarantee. Hence, it is without loss of generality to focus on drop-out attacks.

\subsection{From Robust Shuffle Privacy to Pan-Privacy}
\cite{BalcerCJM20} prove a reduction from robust shuffle privacy to pan-privacy in the context of uniformity testing and counting distinct elements. Here, we note that the technique can be applied to essentially any distributional problem, so we state it as a standalone theorem.  Using this theorem we will be able to obtain lower bounds for the shuffle model from those we prove for the pan-private model.

We begin by establishing some notation. For any universe $\cX$, let $\bU$ denote any fixed distribution over $\cX$. For any distribution $\bP$ over $\cX$ and any $b\in[0,1]$, let $\bP_{(b)}$ denote the mixture $b\cdot \bP+ (1-b)\cdot \bU$.

\begin{thm}[Generalization of \cite{BalcerCJM20}]
\label{thm:shuffle-to-pan}
For any $n$ and any $(\eps,\delta,1/3)$-robustly shuffle private protocol $\Pi$, there exists an $(\eps,\delta)$-pan-private algorithm $M^\Pi$ such that
\begin{equation}
\label{eq:equivalence}
    \dtv(M^\Pi(\bU^{n/3}), \Pi(\bU^n))=0    
\end{equation}
and, for any $\bP$ over $\cX$,
\begin{equation}
\label{eq:dilution}
    \dtv(M^\Pi(\bP^{n/3}), \Pi(\bP^n_{(2/9)}) ) < \exp(-\Omega(n)).  
\end{equation}
In particular, if $n$ is larger than some absolute constant, $\dtv(M^\Pi(\bP^{n/3}), \Pi(\bP^n_{(2/9)}) ) < 1/6$.
\end{thm}

\begin{algorithm}
\caption{$M^\Pi$, an online algorithm built from a shuffle protocol}
\label{alg:shuffle-to-pan}

\KwIn{Data stream $\vec{x}\in \cX^{n/3}$; a shuffle protocol $\Pi=(\Pi_R,\Pi_A)$ that expects $n$ inputs}

Create initial state $S_0 \gets (\Pi_S \circ \Pi_R^{n/3})(\bU^{n/3})$

Sample $N' \sim \Bin(n,2/9)$

Set $N' \gets \min(N',n/3)$

\For{$i\in[n/3]$}{
    \lIf{$i \leq N'$}{ $W_i \gets x_i$}
    \lElse{$W_i \sim \bU$}
    
    Create the state $S_i$ by shuffling the messages from $S_{i-1}$ with those from $\Pi_R(W_i)$
}

Create $\vec{Y}$ by shuffling the messages from $S_{n/3}$ with those from $\Pi_R^{n/3}(\bU^{n/3})$

\Return{$\Pi_A(\vec{Y})$}

\end{algorithm}

\begin{proof}
We present a concise version of $M^\Pi$ in Algorithm \ref{alg:shuffle-to-pan}. Although it does not explicitly take the form specified by Definition \ref{def:online}, it is straightforward to decompose it into a sequence of algorithms $$(M_1,\dots,M_{n/3},M_\Ou).$$

\mypar{Pan-privacy:} For any user $i$ and intrusion time $t$, we prove that $\big(M^\Pi_{\In}(\vec{x}_{\leq t}), M^\Pi_{\Ou}(M^\Pi_{\In}(\vec{x}))\big)$---the adversary's view---is $(\eps,\delta)$-differentially private \emph{conditioned on arbitrary event $N' = n'$}. If $i > n'$, observe that the algorithm is completely independent of $x_i$. Otherwise, we shall leverage the robust privacy of $\Pi$.

We first consider the case where $t < i$. The state observed by the adversary, $S_t$, is independent of $i$ so it will suffice to prove that $M^\Pi_{\Ou}(M^\Pi_{\In}(\vec{x}))$ is differentially private \emph{conditioned on any event $S_t = s_t$}. Note that $M^\Pi_{\Ou}(M^\Pi_{\In}(\vec{x}))$ is obtained by running $\Pi_A$ on the union of $s_t$ and
\begin{equation}
\label{eq:shuffle-to-pan}
(\Pi_S\circ \Pi^{h}_R )(x_{t+1}, \dots, x_i, W_{i+1}, \dots, W_{n/3}, \underbrace{\bU, \dots, \bU}_{n/3~\textrm{terms}} ),    
\end{equation}
where $h = 2n/3 - t \geq n/3$. We can therefore invoke the robust shuffle privacy of $\Pi$.

Now we consider the case where $t \geq i$. Observe that $M^\Pi_{\In}(\vec{x}_{\leq t})$ is equivalent to
$$
(\Pi_S \circ \Pi^h_R)(\underbrace{\bU_\cX, \dots, \bU_\cX}_{n/3~\textrm{terms}},x_1,\dots,x_i, W_{i+1}\dots, W_t),
$$
where $h = n/3 + t > n/3$. We again invoke the robust shuffle privacy of $\Pi$. And, conditioned on any event $M^\Pi_{\In}(\vec{x}_{\leq t}) = s_t$, we argue that $M^\Pi_{\Ou}(M^\Pi_{\In}(\vec{x}))$ is independent of $x_i$. This follows from our previous observation that $M^\Pi_{\Ou}(M^\Pi_{\In}(\vec{x}))$ is obtained by running $\Pi_A$ on the union of $s_t$ and \eqref{eq:shuffle-to-pan}; $x_i$ is not an input to this function.

\mypar{Bound on TV distance:} In the case where the input $\vec{X}$ is drawn from $\bU^{n/3}$, observe that every execution of $\Pi_R$ made by $M^\Pi$ is on an independent sample from $\bU$. Because the output of the algorithm is obtained by running $\Pi_A$ on $n$ such executions, we immediately have $M^\Pi(\bU^{n/3}) = \Pi(\bU^n)$.

Otherwise, consider $n$ samples from $\bP_{(2/9)}$. The number of samples drawn from $\bP$ is distributed as $\Bin(n,2/9)$. By Hoeffding's bound, $\pr{}{\Bin(n,2/9) > n/3} < \exp(-\Omega(n))$. Thus the TV distance between $\Bin(n,2/9)$ and the distribution of $N'$ is at most $\exp(-\Omega(n))$. In turn, the TV distance between
\[
\Pi(\bP^n_{(2/9)})= \Pi_A(\Pi_S(\overbrace{ \underbrace{\Pi_R(\bP), \dots, \Pi_R(\bP)}_{\Bin(n,2/9) ~\textrm{terms}}, \Pi_R(\bU),\dots,\Pi_R(\bU) }^{n~\textrm{terms}}))
\]
and
\[
M^\Pi(\bP^{n/3})= \Pi_A(\Pi_S(\overbrace{ \underbrace{\Pi_R(\bP), \dots, \Pi_R(\bP)}_{N' ~\textrm{terms}}, \Pi_R(\bU),\dots,\Pi_R(\bU) }^{n~\textrm{terms}} ))
\]
is at most $\exp(-\Omega(n))$ as well. This concludes the proof.
\end{proof}
\section{Main Lower Bound} \label{sec:main-lb}

Let $M$ be a pan-private algorithm.  Let $\{\bP_{v}\}_{v \in \cV}$ be a family of distributions, $V$ be uniform over $\cV$, and $\bU = \Ex_{v \sim V}(\bP_{v})$ be the uniform mixture over the distributions.  Let $\bU^n$ be the product distribution consisting of $n$ copies of $\bU$ and let $\bP^{n}_{V} = \Ex_{v \sim V}(\bP_{v}^n)$ be the mixture of product distributions.  Note that $\bU = \bP_{V}^1$.

An important quantity that we will show measures how hard it is for pan-private algorithms to distinguish $\bU^n$ from $\bP_V^n$ is the \emph{$(\infty \mathord{\to} 2)$-norm}\footnote{We call this quantity the $(\infty \mathord{\to} 2)$-norm because it is equal to the better known $(\infty \mathord{\to} 2)$-norm, $\sup_{z} \| M z \|_2 / \|z\|_{\infty}$, of the matrix $M$ defined by $M_{v,x} = \bP_v(x) - U(x)$.} of $\{\bP_{v}\}$, which defined as
$$
\|\{\bP_v\}\|_{\infty \to 2} = \sup_{f \from \cX \to [\pm 1]} \ex{v \sim V}{\left( \ex{x \sim \bP_v}{f(x)} - \ex{x \sim U}{f(x)}  \right)^2}^{1/2}
$$

The main goal of this section is to prove the following theorem.
\begin{thm} \label{thm:main-lb}
	If $\{\bP_v\}_{v \in \cV}$ is a family of distributions and $M$ is an $(\eps,\delta)$-pan private algorithm such that \footnote{We use $x \ll y$ to indicate that $x \leq c y$ for a sufficiently small numerical constant $c > 0$.} $\delta \log\nicefrac {|\cV|}{\delta} \ll \eps^2 \| \{\bP_{v}\} \|_{\infty \to 2}^{2}$ and $\dtv(M(\bP_V^n), M(\bU^n))$ is larger than a positive constant, 
	then
	$$
	n \geq \Omega\left(\frac{1}{\eps \| \{\bP_v\} \|_{\infty \to 2}}\right)
	$$
	More generally,
	$
	n \geq 1/O(\eps \| \{ \bP_v \} \|_{\infty \to 2} + \sqrt{\delta \log \nicefrac{|\cV|}{\delta}})
	$
\end{thm}

The main tool we use to prove Theorem~\ref{thm:main-lb} is the following information inequality.
\begin{lem} \label{lem:main-lb}
For any $(\eps,\delta)$-pan private algorithm $M$,
$$
    \dtv(M(\bP_V^n), M(\bU^n)) \leq n \cdot \sqrt{\tfrac{1}{2}  I_{\eps,\delta}(\{\bP_{v}\})} 
$$
where we define
$
    I_{\eps,\delta}(\{\bP_{v}\}) = \sup_{M \from \cX \to \cR \atop \textrm{$(\eps,\delta)$-DP}} I( M(\bP_V); V)
$
\end{lem}

\begin{proof}[Proof of Lemma~\ref{lem:main-lb}]
As a shorthand, let $\bQ_{i}$ denote the distribution of $M(\bU^{i}, \bP_{V}^{n-i})$. This is the distribution of the algorithm's output on a data stream where the first $i$ elements are sampled i.i.d. from $\bU$ and the rest from $\bP_V$. Note that $\bQ_{0} = M(\bP_{V}^n)$ and $\bQ_{n} = M(\bU^n)$.  By the triangle inequality we have
$$
    \dtv(M(\bP_V^n), M(\bU^n)) = \dtv(\bQ_0, \bQ_n) \leq \sum_{i=1}^{n} \dtv(\bQ_{i-1},\bQ_{i}).
$$
Thus, in order to prove the theorem it is enough to show that for every $i = 1,\dots,n$,
\begin{equation} \label{eq:main-lb-0}
\dtv(\bQ_{i-1},\bQ_{i}) \leq \sqrt{\tfrac{1}{2}  I_{\eps,\delta}(\{\bP_{v}\})}
\end{equation}

\newcommand{\Xn}{X_{i+1 \cdots n}}

Before proving \eqref{eq:main-lb-0}, we give a simplified diagram of the relevant random variables in the two distributions $\bQ_{i-1},\bQ_{i}$ in Figure~\ref{fig:rvs}.  For the purposes of comparing $\bQ_{i-1}$ and $\bQ_{i}$, we can group all of the inputs $X_1,\dots,X_{i-1} \sim \bU^{i-1}$ into one random variable and all of the inputs $\Xn \sim \bP_V^{n-i}$ into another random variable.  Moreover, in $\bQ_{i-1}$, $X_i$ is drawn from $\bP_V$, for the same choice of $V$ as $\Xn$, whereas in $\bQ_{i}$, $X_i$ is drawn from $\bU$.  

\begin{figure}[h!]
	\centering
\begin{tikzpicture}[
node/.style={rectangle, draw=black!100, fill=blue!5, minimum size=7mm},
]
\node[node]		(Xi)														{$X_i$};
\node[node]		(X1)		[left = 5mm of Xi] 		{$X_{1 \cdots i-1}$};
\node[node]		(Xn)        [right = 5mm of Xi]  	{$\Xn$};
\node[node]    	(V)          [below = 5mm of Xi]              {$V$};

\node[node]		(Si)       [above = 5mm of Xi]  {$S_i$};
\node[node]		(Sn)	[right = 5mm of Si]	{$S_n$};

\draw[->] (V.north) -- (Xi);
\draw[->] (V.east) -- (Xn.south);
\draw[->] (Si) -- (Sn);
\draw[->] (Xn.north) -- (Sn.south);
\draw[->] (Xi) -- (Si);
\draw[->] (X1.north) -- (Si.west);
\end{tikzpicture} \qquad\qquad\qquad\qquad
\begin{tikzpicture}[
node/.style={rectangle, draw=black!100, fill=blue!5, minimum size=7mm},
]
\node[node]		(Xi)														{$X_i$};
\node[node]		(X1)		[left = 5mm of Xi] 		{$X_{1 \cdots i-1}$};
\node[node]		(Xn)        [right = 5mm of Xi]  	{$\Xn$};
\node[node]    	(V)          [below = 5mm of Xi]              {$V$};

\node[node]		(Si)       [above = 5mm of Xi]  {$S_i$};
\node[node]		(Sn)	[right = 5mm of Si]	{$S_n$};

\draw[->] (V) -- (Xn.south);
\draw[->] (Si) -- (Sn);
\draw[->] (Xn.north) -- (Sn.south);
\draw[->] (Xi) -- (Si);
\draw[->] (X1.north) -- (Si.west);
\end{tikzpicture}
\caption{A simplified diagram of the relevant random variables in $\bQ_{i-1}$ (left) and $\bQ_{i}$ (right).}
\label{fig:rvs}
\end{figure}
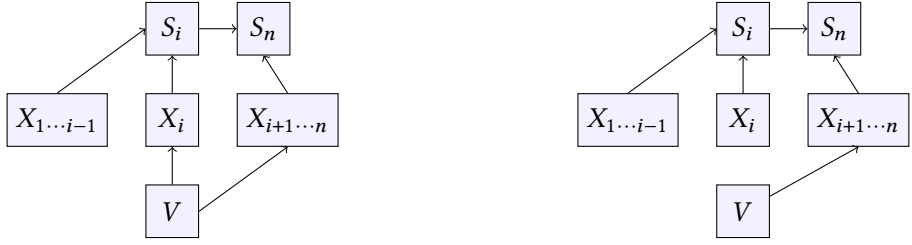

Now, observe that the random variables $S_i$ and $\Xn$ have the same \emph{marginal} distribution in both $\bQ_{i-1},\bQ_{i}$.  However, in $\bQ_{i-1}$ they are correlated by the shared choice of $V$, and in $\bQ_{i}$ they are independent.  Moreover, $S_{n}$ is a post-processing of the pair $(S_i,\Xn)$.  Thus, using $(S_i, \Xn)$ to denote the joint distribution of $S_i(V)$ and $\Xn(V)$ in $\bQ_{i-1}$, and applying the data-processing inequality, we have
\begin{align}
	\dtv(\bQ_{i-1}, \bQ_{i}) 
	\leq{} &\dtv( (S_i,\Xn), (S_i \otimes \Xn) ) \notag \\
	\leq{} &\ex{s_i \sim S_i}{\dtv( \Xn \cond{S_i = s_i}, \Xn)} \tag{Fact~\ref{fact:tvchainrule}}
\end{align}
where the last inequality uses the following fact.
\begin{fact} \label{fact:tvchainrule}
	If $(A,B)$ and $(A,B')$ are joint distributions, $\dtv((A,B),(A,B')) \leq \ex{a \sim A}{\dtv(B \cond{A = a}, B' \cond{A=a})}$.
\end{fact}
Next, since $S_i$ and $\Xn$ are independent conditioned on $V$, we have
\begin{align}
	\ex{s_i \sim S_i}{\dtv( \Xn \cond{S_i = s_i}, \Xn)}
	\leq{} \ex{s_i \sim S_i}{\dtv( V \cond{S_i = s_i}, V)} \tag{Fact~\ref{fact:markovchain}}
\end{align}
where we use the following fact.
\begin{fact} \label{fact:markovchain}
	If $(A,B,C)$ are jointly distributed random variables and $A$ and $B$ are independent conditioned on $C$, then for every $a \in \mathrm{supp}(A)$, $\dtv(B \cond{A = a}, B) \leq \dtv(C \cond{A=a}, C)$.
\end{fact}
We prove Facts \ref{fact:tvchainrule} and \ref{fact:markovchain} in Appendix \ref{sec:supporting-proofs}. From this point we can calculate
\begin{align}
\ex{s_i \sim S_i}{\dtv( V \cond{S_i = s_i}, V)}
\leq{} &\sqrt{\ex{s_i \sim S_i}{\dtv( V \cond{S_i = s_i}, V)^2}} \tag{Jensen's Inequality} \\
\leq{} &\sqrt{\ex{s_i \sim S_i}{\tfrac{1}{2} \cdot \dkl( V \cond{S_i = s_i} \| V)}} \tag{Pinsker's Inequality} \\
={}	 &\sqrt{\ex{s_i \sim S_i}{\tfrac{1}{2} \cdot \dkl( (S_i,V) \| (S_i \otimes V))}} \tag{chain rule for KL-divergence} \\
\leq{} &\sqrt{\tfrac{1}{2} \cdot I(S_i ; V)} \tag{definition of mutual information}
\end{align}

Lastly, we argue that $I(S_i; V) \leq I_{\eps,\delta}(\{\bP_v\})$ using pan-privacy.  The intuition is that pan privacy requires $S_{i}$ to be $(\eps,\delta)$-differentially private as a function of the prefix $X_1,\dots,X_i$.  Moreover, $X_1,\dots,X_{i-1}$ are drawn from the fixed distribution $\bU^{i-1}$ that is independent from $V$.  Therefore, we can fix the distribution of $X_1,\dots,X_{i-1}$ and view $S_{i}$ as an $(\eps,\delta)$-differentially private function of just $X_{i}$.  Specifically, given an $(\eps,\delta)$-pan private algortihm $M$, and $i$, define the function $f_i \from \cX \to \cR$ as follows: $f_i(x)$ samples $X_1,\dots,X_{i-1} \sim \bU^{i-1}$, computes $s_1 = M_1(X_1)$, $s_2 = M_2(X_2,s_1)$, \dots, $s_{i-1} = M_{i-1}(X_{i-1},s_{i-2})$, and outputs $r = M_{i}(x, s_{i-1})$.  Pan-privacy guarantees that $f_{i}(x) = M_{i}(X_1,\dots,X_{i-1},x)$ is $(\eps,\delta)$-differentially private as a function of $x$.  Note that $S_i \cond{X_i = x}$ is distributed identically as $f_{i}(x)$.  Therefore
$$
\sqrt{\tfrac{1}{2} I(S_i ; V)} = \sqrt{\tfrac{1}{2} I(M_i(\bP_V) ; V)} \leq \sqrt{\tfrac{1}{2} I_{\eps,\delta}(\{\bP_v\})}
$$
Combining with the previous calculations gives
$$
\dtv(\bQ_{i-1},\bQ_{i}) \leq \sqrt{\tfrac{1}{2} I_{\eps,\delta}(\{\bP_v\})},
$$
as desired.
\end{proof}

To use Lemma~\ref{lem:main-lb} we need a bound on the mutual information $I_{\eps,\delta}(\{\bP_v\})$.  A result of Duchi, Jordan, and Wainwright~\cite{DuchiJW13}, gives such a bound for the case of $\delta = 0$.
\begin{lem} [\cite{DuchiJW13}] \label{lem:inf-to-2-bound-pure}
	$
	I_{\eps,0}(\{\bP_{v}\}) \leq O(\eps^2 \|\{\bP_{v}\}\|_{\infty \to 2}^2).
	$
\end{lem}
We give a simple extension to the case of $\delta > 0$.
\begin{lem} \label{lem:inf-to-2-bound}
	$
	I_{\eps,\delta}(\{\bP_{v}\}) \leq O(\eps^2 \|\{\bP_{v}\}\|_{\infty \to 2}^2 + \delta \log \nicefrac{|\cV|}{\delta} ).
	$
\end{lem}
Therefore, we will obtain Theorem~\ref{thm:main-lb} as an immediate corollary of Lemma~\ref{lem:main-lb} and Lemma~\ref{lem:inf-to-2-bound}.  The proof of Lemma~\ref{lem:inf-to-2-bound} from Lemma~\ref{lem:inf-to-2-bound-pure} relies on the following statement, which is an easy consequence of a structural result of Kairouz, Oh, and Viswanath~\cite{KairouzOV15}.
\begin{lem} \label{lem:approx-to-pure}
If $M \from \cX \to \cR$ is $(\eps,\delta)$-differentially private, then there is a $(2\eps,0)$-differentially private $M'$ such that
$$
\forall x \in \cX~~\dtv(M(x),M'(x)) \leq \delta
$$
\end{lem}
For completeness, we prove this lemma in Appendix~\ref{sec:supporting-proofs}.

\begin{proof}[Proof of Lemma~\ref{lem:inf-to-2-bound}]
Let $M$ be any $(\eps,\delta)$-differentially private function with input $x\in \cX$.  Lemma~\ref{lem:approx-to-pure} guarantees that there exists a mechanism $M'$ that is $(2\eps,0)$-differentially private and satisfies

$$
\forall x \in \cX~~\dtv(M(x),M'(x)) \leq \delta
$$
In particular, $\dtv(M(\bP_V), M'(\bP_V)) \leq \delta$.  Therefore, there exists a joint distribution $(M,M')$ such that $M = M(\bP_V)$, $M' = M'(\bP_V)$ and $\pr{}{M \neq M'} \leq \delta$.  Let $B$ be the binary random variable $\ind\{M \neq M'\}$.  Thus, there is a joint distribution $(M,M',B)$ such that $(B = 0 \Longrightarrow R = R')$ and $\pr{}{B \neq 0} \leq \delta$.  Therefore,
\begin{align*}
I(V ; R)
\leq{} &I(V ; M,M',B) \\
\leq{}& I(V ; M,M' \mid B) + H(B) \\
={} &I(V ; M,M' \mid B = 0) \pr{}{B = 0} + I(V; M,M' \mid B=1) \pr{}{B=1} + H(B) \\
\leq{} &I(V; M') + H(V) \delta + H(B) \\
={} &I(V;M') +  O(\delta \log|\cV| + \delta \log(1/\delta)) \\
\leq{} &I_{2\eps,0}(\{\bP_{v}\}) + O(\delta \log|\cV| + \delta \log(1/\delta)) \\
={} &O(\eps^2 \|\{\bP_{v}\}\|_{\infty \to 2}^2) + O(\delta \log|\cV| + \delta \log(1/\delta))
\end{align*}
The lemma now follows by rewriting the final expression as $O(\delta \log\nicefrac{|\cV|}{\delta})$.
\end{proof}

\subsection{A Family of Hard Distributions} \label{sec:hard-distributions}

In order to apply Theorem~\ref{thm:main-lb} to a learning or optimization problem, we need a family of distributions $\{\bP_{v}\}$ such that $\| \{\bP_{v}\} \|_{\infty \to 2}$ is small and any accurate algorithm for the problem distinguishes $\bP_V^n$ from $\bU^n$. This subsection describes one such family we will use in most of our lower bound arguments.

Let $\cX = \pmo^d$ be the data domain.  For a parameter $\alpha \in (0,\nicehalf)$, a non-empty set $\ell \subseteq [d]$, and a bit $b \in \pmo^{d}$, we define the distribution $\bP_{d,\ell,b,\alpha}$ to be uniform on $\pmo^{d}$ except biased so that $\mathbb{E}_{x \sim \bP_{d,\alpha, \ell,b,\alpha}}(\prod_{i \in t} x_i) = 2\alpha b$. Its probability mass function is
\begin{equation} \label{eq:Ptb}
    \bP_{d,\ell,b,\alpha}(x) =
    \begin{cases}
        (1+2\alpha) 2^{-d} & \textrm{if $\prod_{i \in t} x_i = b$} \\
        (1-2\alpha) 2^{-d} & \textrm{if $\prod_{i \in t} x_i = -b$}
    \end{cases}
\end{equation}
Note that, by construction, for every non-empty $t' \neq t$, $\mathbb{E}_{x \sim \bP_{d,\ell,b,\alpha}}(\prod_{i \in t'} x_i) = 0$.

For dimension $d$, a parameter $k \leq d$, and $\alpha \in (0,\nicehalf)$, we define the family
\begin{equation} \label{eq:Pdk}
\cP_{d,k,\alpha} = \{ \bP_{d,\ell,b,\alpha} : t \subseteq [d], |t| \in [k], b \in \pmo \}
\end{equation}

\begin{fact}
\label{fact:Pdk-size}
The size of the family $\cP_{d,k,\alpha}$ is $2\cdot \binom{d}{\leq k}$ where $\binom{d}{\leq k} = \sum_{j=1}^k \binom{d}{j}$.
\end{fact}

\begin{fact}
\label{fact:Pdk-mixture}
The uniform mixture over the family $\cP_{d,k,\alpha}$ is uniform over $\cX$.
\end{fact}

The following lemma is implicit in many lower bounds for local differential privacy (e.g.~\cite{DuchiJW13,Ullman18,EdmondsNU20}), although we reprove it here for completeness.
\begin{lem} \label{lem:P-sample-bound}
    For every $d \in \N$, $k \leq d$, and $\alpha \in (0,\nicehalf)$,
    $$
        \| \cP_{d,k,\alpha} \|_{\infty \to 2}^{2} \leq \frac{4\alpha^2}{\binom{d}{\leq k}}
    $$
\end{lem}
\begin{proof}
We begin by expanding the definition of the $(\infty \to 2)$ norm:
\begin{align*}
\| \cP_{d,k,\alpha} \|_{\infty \to 2}^{2} &= \sup_{f \from \cX \to [\pm 1]} \sum_{\bP \in \cP_{d,k,\alpha}} \frac{1}{|\cP_{d,k,\alpha}|} \cdot \left( \ex{x \sim \bP}{f(x)} - \ex{x \sim \bU}{f(x)}  \right)^2 \\
    &= \sup_{f \from \cX \to [\pm 1]} \sum_{t\subseteq [d], |t|\in [k] \atop b\in \{\pm 1\}} \frac{1}{|\cP_{d,k,\alpha}|} \cdot \left( \sum_{x \in \pmo^d } f(x)\cdot ( \bP_{d,\ell,b,\alpha}(x) - \bU(x)) \right)^2 \\
    &= \sup_{f \from \cX \to [\pm 1]} \frac{1}{2\binom{d}{\leq k}} \cdot \sum_{t\subseteq [d], |t|\in [k] \atop b\in \{\pm 1\}} \left( \sum_{x \in \pmo^d } f(x)\cdot ( \bP_{d,\ell,b,\alpha}(x) - \bU(x)) \right)^2 \stepcounter{equation} \tag{\theequation} \label{eq:P-sample-bound-1}
\end{align*}
The final equality comes from Fact \ref{fact:Pdk-size}.
Note that \eqref{eq:Ptb} is equivalent to $\bP_{d,\ell,b,\alpha}(x) = (1 + 2\alpha b \cdot \prod_{i \in t} x_i)2^{-d}$ and, via Fact \ref{fact:Pdk-mixture}, $\bU(x)=2^{-d}$. Thus,
\begin{align*}
\eqref{eq:P-sample-bound-1} &= \sup_{f \from \cX \to [\pm 1]} \frac{1}{2\binom{d}{\leq k}} \cdot  \sum_{t\subseteq [d], |t|\in [k] \atop b\in \{\pm 1\}} \left( \sum_{x \in \{\pm 1\}^d } f(x)\cdot 2\alpha b \cdot \prod_{i \in t} x_i \cdot 2^{-d} \right)^2  \\
    &= \sup_{f \from \cX \to [\pm 1]} \frac{2\alpha^2}{\binom{d}{\leq k}} \cdot \sum_{t\subseteq [d], |t|\in [k] \atop b\in \{\pm 1\}} \left( \sum_{x \in \{\pm 1\}^d } f(x) \cdot \prod_{i \in t} x_i \cdot 2^{-d} \right)^2 \\
    &= \sup_{f \from \cX \to [\pm 1]} \frac{4\alpha^2}{\binom{d}{\leq k}} \cdot \sum_{t\subseteq [d], |t|\in [k]} \left( \sum_{x \in \{\pm 1\}^d } f(x) \cdot \prod_{i \in t} x_i \cdot 2^{-d} \right)^2 \\
    &\leq \sup_{f \from \cX \to [\pm 1]} \frac{4\alpha^2}{\binom{d}{\leq k}} \cdot \sum_{t\subseteq [d]} \left( \sum_{x \in \{\pm 1\}^d } f(x) \cdot \prod_{i \in t} x_i \cdot 2^{-d} \right)^2 \stepcounter{equation} \tag{\theequation} \label{eq:P-sample-bound-2}
\end{align*}
Define $\hat{f}(t) := \ex{X\sim \bU}{ f(X)\cdot \prod_{i\in t} X_i}$, the Fourier transform over the Boolean hypercube. This is precisely the term being squared above. So we have
\begin{align*}
\eqref{eq:P-sample-bound-2} &= \frac{4\alpha^2}{\binom{d}{\leq k}} \cdot \sup_{f \from \cX \to [\pm 1]} \sum_{t\subseteq [d]} \hat{f}(t)^2 \\
    &= \frac{4\alpha^2}{\binom{d}{\leq k}} \cdot \sup_{f \from \cX \to [\pm 1]} \ex{X\sim \bU}{f(X)^2} \tag{Parseval's identity} \\
    &\leq \frac{4\alpha^2}{\binom{d}{\leq k}}
\end{align*}
This concludes the proof.
\end{proof}

The following is an immediate corollary of Theorem \ref{thm:main-lb}, Lemma \ref{lem:P-sample-bound}, and Fact \ref{fact:Pdk-size}.

\begin{thm}
\label{thm:P-lb}
Let $\bP_{d,L,B,2\alpha}$ denote a distribution chosen uniformly at random from $\cP_{d,k,\alpha}$ (where $L$ is a uniformly random subset of $[d]$ with size $\leq k$ and $B$ is a uniformly random member of $\pmo$). If $M$ is an $(\eps,\delta)$-pan private algorithm such that  $\delta \log\nicefrac {\binom{d}{\leq k}}{\delta} \ll \alpha^2 \eps^2  / \binom{d}{\leq k} $ and $\dtv(M(\bP_{d,L,B,\alpha}^n), M(\bU^n))$ is larger than a positive constant, 
	then
	$$
	n \geq \Omega\left(\frac{ \sqrt{\binom{d}{\leq k}} }{ \alpha \eps }\right)
	$$
\end{thm}

\section{Lower Bounds for Simple Hypothesis Testing}

In this section, we use Theorem \ref{thm:P-lb} obtain lower bounds for the problem of simple hypothesis testing. We first prove a lower bound that holds under pan-privacy, then adapt it for robust shuffle privacy via Theorem \ref{thm:shuffle-to-pan}. This pattern is repeated in the subsequent lower bound sections.

\begin{defn}[$d$-Wise Simple Hypothesis Testing]
Let $d$ be any integer larger than 1 and let $\alpha$ be any real in the interval $(0,\nicehalf)$. An algorithm $M$ solves \emph{$d$-wise simple hypothesis testing with error $\alpha$ and sample complexity $n$} if, for any set of $d$ distributions $\cP$ satisfying  $\dtv(\bP,\bP')\geq \alpha$ for every distinct pair $\bP,\bP' \in \cP$, when given $n$ independent samples from an arbitrary $\bP \in \cP$ as input, the algorithm outputs $\bP$ with probability $\geq 99/100$. This probability is over the randomness of the samples and of $M$.
\end{defn}




\begin{thm}
\label{thm:hyp-selection-pan}
If $M$ is an $(\eps,\delta)$-pan-private algorithm that solves $d$-wise simple hypothesis testing with error $\alpha$ and $\delta \log \nicefrac{d}{\delta} \ll \alpha^2 \eps^2 / d$, then its sample complexity is $n=\Omega(\sqrt{d}/\alpha\eps)$.
\end{thm}
\begin{proof}
Consider the set of distributions $\{\bU\} \cup \cP_{d,1,\alpha}$.  Note that this is a family of $2d+1$ distributions.  From Fact \ref{fact:Pdk-size}, its size is $2d+1$. We also prove the following in the Appendix:
\begin{clm}
\label{clm:P-separated}
For any $\bP\neq\bP' \in \{\bU\} \cup \cP_{d,1,\alpha}$, $\dtv(\bP,\bP') \geq \alpha$.
\end{clm}
The upshot is that $\{\bU\} \cup \cP_{d,1,\alpha}$ is a valid set of distributions for $(2d+1)$-wise hypothesis testing. We now argue that the accuracy of $M$ for this problem instance implies that we can invoke Theorem \ref{thm:P-lb}.

To do so, let $\bP_{d,L,B,\alpha}$ denote a distribution chosen uniformly at random from $\cP_{d,1,\alpha}$. We show that the total variation distance between $M(\bU^n)$ and $M(\bP^n_{d,L,B,\alpha})$ is at least some positive constant.
\begin{align*}
&\dtv(M(\bU^n),M(\bP^n_{d,L,B,\alpha}))\\
    ={}& \max_{\cP \subseteq \{\bU \}\cup \cP_{d,k,\alpha}} \left|\pr{}{M(\bU^n) \in \cP} - \pr{}{M(\bP^n_{d,L,B,\alpha}) \in \cP} \right| \\
    \geq{}& \pr{}{M(\bU^n) \in \{ \bU \} } - \pr{}{M(\bP^n_{d,L,B,\alpha}) \in \{ \bU \} } \\
    \geq{}& \pr{}{M(\bU^n) \in \{ \bU \} } - \frac{1}{100} \\
    \geq{}& \frac{99}{100} - \frac{1}{100} = \frac{49}{50}
\end{align*}
To obtain the second inequality, we first observe that $\bP_{d,t,b,\alpha} \neq \bU$ for every $t,b$ so $\bU$ would be an incorrect output. Then we use the fact that $M$ solves simple hypothesis testing: it is incorrect with probability at most $1/100$. The same reasoning yields the third inequality.

From Theorem \ref{thm:P-lb}, we conclude that $n = \Omega\left(\frac{1}{\eps \| \cP_{d,1,\alpha} \|_{\infty \to 2}}\right) = \Omega \left( \sqrt{d}/\alpha \eps \right).$.  This lower bound holds for a family of $2d+1$ distributions, so the claimed result follows by rescaling $d$.
\end{proof}

The next theorem adapts our proof to the robust shuffle privacy setting:

\begin{thm} \label{thm:hyp-selection-shuffle}
If $\Pi$ is an $(\eps,\delta,1/3)$-robustly shuffle private protocol that solves $d$-wise simple hypothesis testing with error $\alpha$ and $\delta\log \nicefrac{d}{\delta} \ll \alpha^2\eps^2 / d$, then its sample complexity is $n=\Omega(\sqrt{d}/\alpha\eps)$.
\end{thm}
\begin{proof}
As before, let $\bP_{d,L,B,\alpha}$ denote a distribution chosen uniformly at random from $\cP_{d,1,\alpha}$. Let $\Pi$ denote an algorithm in the shuffle model that solves $(2d+1)$-wise simple hypothesis testing with accuracy $2\alpha/9$.

Let $M^\Pi$ denote the $(\eps,\delta)$-pan-private algorithm guaranteed by Theorem \ref{thm:shuffle-to-pan}. We will lower bound the total variation distance between $M^\Pi(\bU^{n/3})$ and $M^\Pi(\bP^{n/3}_{d,L,B,\alpha})$.

\begin{align*}
&\dtv(M^\Pi(\bU^{n/3}), M^\Pi(\bP^{n/3}_{d,L,B,\alpha})) \\
\geq{}& \pr{}{M^\Pi(\bU^{n/3}) \in \{ \bU \} } - \pr{}{M^\Pi(\bP^{n/3}_{d,L,B,\alpha}) \in \{ \bU \} } \\
\geq{}& \pr{}{\Pi(\bU^n) \in \{ \bU \}} - \pr{}{\Pi(\bP^n_{d,L,B,2\alpha/9}) \in \{ \bU \} } - \frac{1}{6} \tag{Theorem \ref{thm:shuffle-to-pan}}\\
\geq{}& \frac{49}{50} - \frac{1}{6} 
= \frac{61}{75}
\end{align*}
The third inequality comes from repeating the analysis in the proof of Theorem \ref{thm:hyp-selection-pan}. Since $M^\Pi$ is an $(\eps,\delta)$-pan-private algorithm such that $$\dtv(M^\Pi(\bU^{n/3}),M^\Pi(\bP^{n/3}_{d,L,B,\alpha}))$$ is at least a positive constant, we invoke Theorem \ref{thm:P-lb} to conclude that $n=\Omega(\sqrt{d}/\alpha\eps)$.  The claimed theorem follows by rescaling $\alpha$ and $d$.
\end{proof}

\section{Lower Bounds for Sparse Mean Estimation}

\begin{defn}
Let $\alpha$ be any real in the interval $(0,\nicehalf)$ and let $k\leq d$ be any integers larger than 1. An algorithm $M$ \emph{solves $(d,k,\alpha)$-sparse mean estimation with sample complexity $n$} if, for any distribution $\bP$ over $\pmo^d$ whose mean $\vec{\mu}$ satisfies $\|\vec{\mu} \|_0\leq k$, it receives $n$ independent samples from $\bP$ as input and outputs $\vec{V}\in[-1,+1]^d$ such that
$\|\vec{\mu}-\vec{V}\|_\infty \leq \alpha$ with probability at least $99/100$. This probability is taken over the randomness of the samples observed by $M$ and $M$ itself.
\end{defn}

\begin{thm} \label{thm:sparse-mean-pan}
If $M$ is an $(\eps,\delta)$-pan-private algorithm that solves $(d,1,\alpha)$-sparse mean estimation and $\delta \log \nicefrac{d}{\delta} \ll \alpha^2 \eps^2 / d$, then its sample complexity is $n=\Omega(\sqrt{d}/\alpha\eps)$.
\end{thm}
\begin{proof}
As before, let $\bP_{d,L,B,\alpha}$ denote a distribution chosen uniformly at random from $\cP_{d,1,\alpha}$.  By construction, the mean of this distribution is 1-sparse, namely it is $B\cdot \vec{e}_{L}$ wehre $\vec{e}_{L}$ is the $L$-th standard basis vector. We show that the total variation distance between $M(\bU^n)$ and $M(\bP^n_{d,L,B,\alpha})$ is at least a constant. This time, we argue that the former is more likely to output a ``small'' vector than the latter. Specifically,
\begin{align*}
&\dtv(M(\bU^n),M(\bP^n_{d,L,B,\alpha}))\\
\geq{}& \pr{}{\| M(\bU^n) \|_\infty \leq \alpha} - \pr{}{\| M(\bP^n_{d,L,B,\alpha}) \|_\infty \leq \alpha} \\
={}& \pr{}{\| M(\bU^n) - \ex{}{\bU} \|_\infty \leq \alpha} - \pr{}{\| M(\bP^n_{d,L,B,\alpha}) \|_\infty \leq \alpha} \tag{$\ex{}{\bU}=\vec{0}$}\\
\geq{}& \frac{99}{100} - \pr{}{\| M(\bP^n_{d,L,B,\alpha}) \|_\infty \leq \alpha}\\
\geq{}& \frac{99}{100} - \pr{}{\| M(\bP^n_{d,L,B,\alpha}) - \ex{}{\bP_{d,L,B,\alpha}} \|_\infty > \alpha} \tag{$\norm{\ex{}{\bP_{d,L,B,\alpha}}}{\infty} = 2\alpha$}\\
\geq{}& \frac{99}{100} - \frac{1}{100}
= \frac{49}{50}
\end{align*}
From Theorem \ref{thm:P-lb}, we conclude that $n = \Omega \left( \sqrt{d}/\alpha \eps \right).$
\end{proof}

The next theorem adapts our proof to the robust shuffle privacy setting:

\begin{thm} \label{thm:sparse-mean-shuffle}
If $\Pi$ is an $(\eps,\delta,1/3)$-robustly shuffle private protocol that solves $(d, 1, \alpha)$-sparse mean estimation and $\delta\log\nicefrac{d}{\delta} \ll \alpha^2\eps^2 / d$, then its sample complexity is $n=\Omega(\sqrt{d}/\alpha\eps)$.
\end{thm}
\begin{proof}
As before, let $\bP_{d,L,B,\alpha}$ denote a distribution chosen uniformly at random from $\cP_{d,1,\alpha}$.  Assume $\Pi$ is a shuffle-model protocol that solves $(d,1,2\alpha/9)$-sparse mean estimation.  We show that $M^\Pi$ distinguishes between $\bU^{n/3}$ and $\bP^{n/3}_{d,L,B,\alpha}$.

\begin{align*}
&\dtv(M^\Pi(\bU^{n/3}),M^\Pi(\bP^{n/3}_{d,L,B,\alpha}))\\
\geq{}& \pr{}{\| M^\Pi(\bU^{n/3}) \|_\infty \leq 2\alpha/9} - \pr{}{\| M^\Pi(\bP^{n/3}_{d,L,B,\alpha}) \|_\infty \leq 2\alpha/9} \\
\geq{}& \pr{}{\| \Pi(\bU^n) \|_\infty \leq 2\alpha/9} - \pr{}{\| \Pi(\bP^{n}_{d,L,B,2\alpha/9}) \|_\infty \leq 2\alpha/9} - \frac{1}{6} \tag{Theorem \ref{thm:shuffle-to-pan}}\\
\geq{}& \frac{49}{50} - \frac{1}{6}
= \frac{61}{75}
\end{align*}
The third inequality comes from repeating the analysis in the proof of Theorem \ref{thm:sparse-mean-pan}. As before, we invoke Theorem \ref{thm:P-lb} to conclude that $n=\Omega(\sqrt{d}/\alpha\eps)$.  The claimed theorem follows from rescaling $\alpha$ and $d$.
\end{proof}
\section{Lower Bounds for Releasing Parity Functions}

\begin{defn}
Let $\alpha$ be any real in the interval $(0,\nicehalf)$ and let $k\leq d$ be any integers larger than 1. An algorithm $M$ \emph{releases width-$k$ parities with error $\alpha$ and sample complexity $n$} if it takes $n$ independent samples from a distribution $\bP$ over $\pmo^d$ and reports a function $F:2^{[d]} \to \R$ such that
$$
\pr{\vec{X} \sim \bP^n \atop F \sim M(\vec{X})}{\forall \ell \subseteq [d],|\ell|\leq k ~ \left| F(\ell) - \ex{x \sim \bP}{\prod_{j \in \ell} x_j} \right| \leq \alpha} \geq 99/100.
$$
\end{defn}

\begin{thm} \label{thm:estimate-parities-pan}
If $M$ is an $(\eps,\delta)$-pan-private algorithm that releases width-$k$ parities with error $\alpha$ and $\delta \log \nicefrac{ \binom{d}{\leq k}}{\delta} \ll \alpha^2 \eps^2 / \binom{d}{\leq k}$, then its sample complexity is $n=\Omega( \sqrt{\binom{d}{\leq k}}/\alpha\eps)$.
\end{thm}

\begin{proof}
Analogous to the previous proofs, let $\bP_{d,L,B,\alpha}$ denote a distribution chosen uniformly at random from the family $\cP_{d,k,\alpha}$. We show that the total variation distance between $M(\bU^n)$ and $M(\bP^n_{d,L,B,\alpha})$ is at least a constant. This time, we argue that the former is more likely to output a function bounded by $\alpha$ than the latter. Specifically,
\begin{align*}
&\dtv(M(\bU^n),M(\bP^n_{d,L,B,\alpha}))\\
\geq{}& \pr{F\sim M(\bU^n)}{\forall \ell \subseteq[d], |\ell|\leq k ~ |F(\ell)|  \leq \alpha} - \pr{ F\sim M(\bP^n_{d,L,B,\alpha})}{\forall \ell \subseteq[d], |\ell|\leq k ~ |F(\ell)|  \leq \alpha} \\
\geq{}& \frac{99}{100} - \pr{ F\sim M(\bP^n_{d,L,B,\alpha})}{\forall \ell \subseteq[d], |\ell|\leq k ~ |F(\ell)|  \leq \alpha} \stepcounter{equation} \tag{\theequation} \label{eq:parity-release-1} \\
\geq{}& \frac{99}{100} - \pr{ F\sim M(\bP^n_{d,L,B,\alpha})}{\forall \ell \subseteq[d], |\ell|\leq k ~ |F(\ell)-2\alpha| > \alpha} \\
\geq{}& \frac{99}{100} - \frac{1}{100} \stepcounter{equation} \tag{\theequation} \label{eq:parity-release-2} \\
={}& \frac{49}{50}
\end{align*}
Inequality \eqref{eq:parity-release-1} follows from the fact that $\forall \ell,b \ex{x \sim \bU}{\prod_{j \in \ell} x_j} = 0$ and the correctness of $M$. Meanwhile \eqref{eq:parity-release-2} follows from the fact that $\forall \ell,b \ex{x \sim \bP_{d,\ell,b,\alpha}}{\prod_{j \in \ell} x_j} = 2\alpha b$ and the correctness of $M$.
From Theorem \ref{thm:P-lb}, we conclude the claimed lower bound on $n$.
\end{proof}

The next theorem adapts our proof to the robust shuffle privacy setting:

\begin{thm} \label{thm:estimate-parities-shuffle}
If $\Pi$ is an $(\eps,\delta,1/3)$-robustly shuffle private protocol that releases width-$k$ parities with error $\alpha$ and $\delta \log \nicefrac{\binom{d}{\leq k}}{\delta} \ll \alpha^2\eps^2 / \binom{d}{\leq k}$, then its sample complexity is $n=\Omega(\sqrt{\binom{d}{\leq k}} / \alpha \eps )$.
\end{thm}
\begin{proof}
Again, let $M^\Pi$ denote the $(\eps,\delta)$-pan-private algorithm given by Theorem \ref{thm:shuffle-to-pan}. We show that $M^\Pi$ distinguishes between $\bU^{n/3}$ and $\bP^{n/3}_{d,L,B,\alpha}$.

\begin{align*}
&\dtv(M^\Pi(\bU^{n/3}), M^\Pi(\bP^{n/3}_{d,L,B,\alpha}))\\
\geq{}& \pr{F\sim M^\Pi(\bU^{n/3})}{\forall \ell \subseteq[d], |\ell|\leq k ~ |F(\ell)|  \leq 2\alpha/9} - \pr{ F\sim M^\Pi(\bP^{n/3}_{d,L,B,\alpha})}{\forall \ell \subseteq[d], |\ell|\leq k ~ |F(\ell)|  \leq 2\alpha/9} \\
\geq{}& \pr{F\sim \Pi(\bU^{n})}{\forall \ell \subseteq[d], |\ell|\leq k ~ |F(\ell)|  \leq 2\alpha/9} - \pr{ F\sim \Pi(\bP^{n}_{d,L,B,2\alpha/9})}{\forall \ell \subseteq[d], |\ell|\leq k ~ |F(\ell)|  \leq 2\alpha/9} -\frac{1}{6} \\
\geq{}& \frac{49}{50} - \frac{1}{6}
= \frac{61}{75}
\end{align*}

The third inequality comes from repeating the analysis in the proof of Theorem \ref{thm:estimate-parities-pan}. As before, we invoke Theorem \ref{thm:P-lb} to conclude the claimed lower bound on $n$.
\end{proof}



\section{Lower Bounds for Selection}

\begin{defn}[Selection]
Let $\alpha$ be any real in the interval $(0,\nicehalf)$ and let $d$ be any integer larger than 1. An algorithm $M$ \emph{solves $(\alpha,d)$-selection with sample complexity $n$} if, for any distribution $\bP$ over $\{\pm 1\}^d$, it takes $n$ independent samples from $\bP$ and selects a coordinate $J \in [d]$ such that $\ex{X\sim \bP}{X_J} \geq \max_j \ex{X\sim \bP}{X_j} - \alpha$ with probability at least $99/100$. This probability is taken over the randomness of the samples observed by $M$ and $M$ itself.
\end{defn}

\begin{thm} \label{thm:selection-pan}
If $M=(M_1,\dots,M_n,M_\Ou)$ is an $(\eps,\delta)$-pan-private algorithm that solves $(\alpha,d)$-selection and $\delta \log \nicefrac{d}{\delta} \ll \alpha^2 \eps^2 / d$, then its sample complexity is $n=\Omega(\sqrt{d}/\alpha\eps)$.
\end{thm}

\begin{proof}
Let $\bP_{d,L,B,\alpha}$ denote a distribution chosen uniformly at random from $\cP_{d,1,\alpha}$. We will again use Theorem \ref{thm:P-lb} but this time our proof will not use $M$ as-is. Instead, we show that $M$ implies another $(\eps,\delta)$-pan-private algorithm $M'$ where the total variation distance between $M'(\bU^n)$ and $M'(\bP^n_{d,L,B,\alpha})$ is at least a positive constant.

Let $\Rad(\alpha)$ be the distribution over $\{\pm 1\}$ with mean $\alpha$. For any $i\in[n]$, define $M'_i$ to be the internal update algorithm that does the following on input $x_i$:
\begin{enumerate}
    \item Draw independent sample $Y_i$ from $\Rad(\alpha)$
    \item $W_i \gets (x_{i,1}, x_{i,2},\dots, x_{i,d}, Y_i)$
    \item Output $M_i(W_i,s_{i-1})$ if $i>1$ else $M_1(W_1)$
\end{enumerate}

$M'$ is the online algorithm defined by $(M'_1,\dots,M'_n,M_\Ou)$. It is $(\eps,\delta)$-pan-private by virtue of using $M$, so it remains to lower bound the TV distance between $M'(\bU^n)$ and $M'(\bP^n_{d,L,B,\alpha})$.
\begin{align*}
&\dtv(M'(\bU^n), M'(\bP^n_{d,L,B,\alpha}))\\
\geq{}& \pr{}{M'(\bU^n)= d+1} - \pr{}{M'(\bP^{n}_{d,L,B,\alpha}) = d+1} \\
={}& \pr{}{M(\bP^n_{d+1,\{d+1\},+1,\alpha/2})= d+1} - \pr{}{M'(\bP^{n}_{d,L,B,\alpha}) = d+1} \stepcounter{equation} \tag{\theequation} \label{eq:selection-pan-1} \\
\geq{}& \frac{99}{100} - \pr{}{M'(\bP^{n}_{d,L,B,\alpha}) = d+1} \stepcounter{equation} \tag{\theequation} \label{eq:selection-pan-2}
\end{align*}
To obtain \eqref{eq:selection-pan-1}, observe that $M'$ feeds into $M$ a stream of $n$ i.i.d. samples from a product distribution where the $(d+1)$-th coordinate has mean $\alpha$, while the rest have mean 0. In our notation, this product distribution is $\bP_{d+1,\{d+1\},+1,\alpha/2}$. Meanwhile, the inequality in \eqref{eq:selection-pan-2} follows from the fact that $M$ solves $(\alpha,d+1)$-selection.

We now upper bound the probability in \eqref{eq:selection-pan-2}. 
\begin{align*}
&\pr{}{M'(\bP^{n}_{d,L,B,\alpha}) = d+1} \\
={}& \pr{}{M'(\bP^{n}_{d,L,B,\alpha}) = d+1,~ B=-1} + \pr{}{M'(\bP^{n}_{d,L,B,\alpha}) = d+1,~ B=+1} \\
\leq{}& \half + \pr{}{M'(\bP^{n}_{d,L,B,\alpha}) = d+1,~ B=+1} \\
={}& \half + \sum_{j=1}^d \pr{}{M'(\bP^{n}_{d,\{j\},+1,\alpha}) = d+1} \cdot \pr{}{T=\{j\}, B=+1} \stepcounter{equation} \tag{\theequation} \label{eq:selection-pan-3}
\end{align*}
We focus our attention on the first term in the product. Observe that $M'$ feeds to $M$ a stream of $n$ iid samples drawn from a distribution where coordinate $j \in [d]$ has mean $2\alpha$, coordinate $d+1$ has mean $\alpha$, and every other coordinate has mean $0$. Here, $j$ is the correct answer to $(\alpha,d+1)$ selection; since $M$ solves $(\alpha,d+1)$-selection, $\pr{}{M'(\bP^{n}_{d,\{j\},+1,\alpha}) = d+1} \leq \nicefrac{1}{100}$. As a result,
\begin{align*}
\eqref{eq:selection-pan-3} &\leq \half + \sum_{j=1}^d \frac{1}{100} \cdot \pr{}{L=\{j\}, B=+1}
    = \half + \frac{1}{100}
    = \frac{51}{100}
\end{align*}

Thus, $\dtv(M'(\bU^n), M'(\bP^n_{d,L,B,\alpha})) \geq \nicefrac{99}{100} - \nicefrac{51}{100} = \nicefrac{12}{25}$. From Theorem \ref{thm:P-lb}, we conclude that $$n = \Omega\left(\frac{1}{\eps \| \cP_{d,1,\alpha} \|_{\infty \to 2}}\right) = \Omega \left( \sqrt{d}/\alpha \eps \right).$$
The claimed theorem now follows by rescaling $d$.
\end{proof}

The next theorem adapts our proof to the robust shuffle privacy setting:

\begin{thm} \label{thm:selection-shuffle}
If $\Pi$ is an $(\eps,\delta,1/3)$-robustly shuffle private protocol that solves $(\alpha, d)$-selection and $\delta\log\nicefrac{d}{\delta} \ll \alpha^2\eps^2 / d$, then its sample complexity is $n=\Omega(\sqrt{d}/\alpha\eps)$.
\end{thm}
\begin{proof}
As before, let $\bP_{d,L,B,\alpha}$ denote a distribution chosen uniformly at random from $\cP_{d,1,\alpha}$. Let $M^\Pi$ denote the $(\eps,\delta)$-pan-private algorithm given by Theorem \ref{thm:shuffle-to-pan}. Like the preceding proof, we show that $M^\Pi$ implies an $(\eps,\delta)$-pan-private algorithm $M'$ that distinguishes between $\bU^{n/3}$ and $\bP^{n/3}_{d,L,B,\alpha}$. We construct $M'$ essentially identically, the differences being that we have $n/3$ instead of $n$ internal algorithms.

To bound the total variation distance between $M'(\bU^{n/3})$ and $M'(\bP^{n/3}_{d,L,B,\alpha})$ we follow the same steps as in the proof of Theorem \ref{thm:selection-pan} except we need to account for the reduction from robust shuffle privacy to pan-privacy (Theorem \ref{thm:shuffle-to-pan})
\begin{align*}
&\dtv(M'(\bU^{n/3}), M'(\bP^{n/3}_{d,L,B,\alpha}))\\
\geq{}& \pr{}{M'(\bU^{n/3})= d+1} - \pr{}{M'(\bP^{n/3}_{d,L,B,\alpha}) = d+1} \\
={}& \pr{}{M^\Pi(\bP^{n/3}_{d+1,\{d+1\},+1,\alpha/2})= d+1} - \pr{}{M'(\bP^{n/3}_{d,L,B,\alpha}) = d+1} \\
\geq{}& \pr{}{\Pi(\bP^{n}_{d+1,\{d+1\},+1,\alpha/9})= d+1} -\frac{1}{6} - \pr{}{M'(\bP^{n/3}_{d,L,B,\alpha}) = d+1} \tag{Theorem \ref{thm:shuffle-to-pan}} \\
\geq{}& \frac{99}{100} - \frac{1}{6} - \pr{}{M'(\bP^{n/3}_{d,L,B,\alpha}) = d+1} \\
\geq{}& \frac{99}{100} - \frac{1}{6} - \left( \half + \frac{1}{100} + \frac{1}{6} \right) \tag{Theorem \ref{thm:shuffle-to-pan}}\\
={}& \frac{11}{75}
\end{align*}

As before, we invoke Theorem \ref{thm:P-lb} to conclude that $n=\Omega(\sqrt{d}/\alpha\eps)$.  The claimed theorem follows from rescaling $\alpha$ and $d$.
\end{proof}

\section{Lower Bounds for Learning Signed Parity Functions}
In this section, we take $\cX = \pmo^{d+1}$ and interpret the bits at index $d+1$ to be labels of the strings. Our focus will be on signed parity functions: given a tuple $(\ell,b) \in 2^{[d]} \times \pmo$ and a string $x\in\cX$, we would like labels to predict the value $b\cdot \prod_{j \in \ell} x_j$. Specifically, for any distribution $\bP$ over $\cX$, we define error function $$\mathrm{err}_\bP(\ell,b) := \pr{X\sim \bP}{b\cdot \prod_{j \in \ell} X_j \neq X_{d+1}},$$ to be the probability of misclassifying a random test example.

\begin{defn}
Let $\alpha \in (0,\half)$ be a parameter and let $1 \leq k \leq d$ be integers.  An algorithm $M$ \emph{learns width-$k$ signed parities with error $\alpha$ and sample complexity $n$} if it takes $n$ independent samples from a distribution $\bP$ over $\cX$ and reports a tuple $(L,B)\in 2^{[d]} \times \pmo$ such that, with probability at least $99/100$, 
$$\mathrm{err}_\bP(L,B) < \min_{\ell,b} \mathrm{err}_\bP(\ell,b) + \alpha.$$ 
This probability is taken over the randomness of the samples and over $M$.
\end{defn}

For this problem, we will use a variant of our family of distributions: for a parameter $\alpha \in [0,\nicehalf]$, a set $\ell \subseteq [d]$, and a bit $b \in \pmo$, we define the distribution $\bQ_{d,\ell,b,\alpha}$ to have probability mass function
\begin{equation} \label{eq:Qtb}
    \bQ_{d,\ell,b,\alpha}(x) =
    \begin{cases}
        (1+2\alpha) 2^{-d-1} & \textrm{if $b\cdot \prod_{j \in \ell} x_j = x_{d+1}$} \\
        (1-2\alpha) 2^{-d-1} & \textrm{if $b\cdot \prod_{j \in \ell} x_j = -x_{d+1}$}
    \end{cases}
\end{equation}

\begin{fact}
\label{fact:signed-parity}
For any $(\ell',b') \neq (\ell,b)$,
\begin{align*}
\pr{X\sim \bQ_{d,\ell,b,\alpha} }{b\cdot \prod_{j \in \ell} X_j = X_{d+1}} &= \half + \alpha \\
\pr{X\sim \bQ_{d,\ell,b,\alpha} }{b'\cdot \prod_{j \in \ell'} X_j = X_{d+1}} &\leq \half
\end{align*}
\end{fact}

For dimension $d$, a parameter $k \leq d$, and $\alpha \in [0,\nicehalf]$, we define the family
\begin{equation} \label{eq:Qdk}
\cQ_{d,k,\alpha} = \{ \bQ_{d,\ell,b,\alpha} : \ell \subseteq [d], |\ell| \leq k, b \in \pmo \}
\end{equation}

\begin{fact}
\label{fact:Qdk-size}
The size of the family $\cQ_{d,k,\alpha}$ is $2\binom{d}{\leq k}+2$.
\end{fact}

\begin{fact}
\label{fact:Qdk-mixture}
The uniform mixture of the family $\cQ_{d,k,\alpha}$ is uniform over $\cX$.
\end{fact}

\begin{lem}
\label{lem:Q-sample-bound}
For every $d \in \N$, $k \leq d$, and $\alpha \in [0,\nicehalf]$,
    $$
        \| \cQ_{d,k,\alpha} \|_{\infty \to 2}^{2} \leq \frac{4\alpha^2}{\binom{d}{\leq k}}
    $$
\end{lem}

For brevity, we defer the proof to the Appendix.

\begin{thm} \label{thm:learning-parities-pan}
If $M=(M_1,\dots,M_n,M_O)$ is an $(\eps,\delta)$-pan-private algorithm that learns width-$k$ signed parities with error $\alpha$ and $\delta \log  \nicefrac{\binom{d}{\leq k}}{\delta} \ll \alpha^2 \eps^2 / \binom{d}{\leq k}$, then its sample complexity is $n=\Omega( \sqrt{\binom{d}{\leq k}}/\alpha\eps)$.
\end{thm}

\begin{algorithm}[t]
\caption{$M'$, an online algorithm }
\label{alg:learning-parities-pan}

\KwIn{Data stream $\vec{x}\in \cX^m$; access to online algorithm $M:\cX^n \to 2^{[d]} \times \pmo$}
\KwOut{A random variable $Z \in \R$}

$S_1 \gets M_1(x_1)$

\For{$i \in [2,n]$}{
    $S_i \gets M_i(x_i, S_{i-1})$
}

\For{$i \in [n+1, m]$}{
    \If{$i=n+1$}{
        $(\hat{L},\hat{B}) \gets M_\Ou(S_n)$

        $C \sim \Lap(1/\eps)$
    }
    \Else{ $(\hat{L},\hat{B}, C) \gets S_{i-1}$ }

    \If{$\prod_{j \in \hat{L}} x_{i,j} = x_{i, d+1} \cdot \hat{B}$}{
        $C \gets C +1$
    }
    
    $S_i \gets (\hat{L},\hat{B}, C)$
}

$L\sim \Lap(1/\eps)$

\Return{$Z\gets C + L$}
\end{algorithm}

\begin{proof}
Analogous to previous proofs, let $\bQ_{d,L,B,\alpha}$ denote a distribution chosen uniformly at random from $\cQ_{d,k,\alpha}$. We argue that $M$ implies an $(\eps,\delta)$-pan-private algorithm $M'$ which takes $m = n + \Theta(1/\alpha \eps)$ values from $\cX$ as input and outputs a real number such that $\dtv(M'(\bU^m), M'(\bQ^m_{d,L,B,\alpha}) )$ is larger than a constant.

We specify $M'$ in Algorithm \ref{alg:learning-parities-pan}. Although it does not explicitly have the structure in Definition \ref{def:online}, it is straightforward to decompose it into a sequence of algorithms. At a high level, $M'$ has a training and a testing phase. In the training phase, it will execute $M$ on the first $n$ samples to obtain a signed parity function $(\hat{L},\hat{B})$. In the testing phase, $M'$ will evaluate the function on the remaining samples and maintain a pan-private estimate of the number of correct predictions. If the samples are drawn from $\bU$, then any choice of parity function makes a correct prediction with only $\nicehalf$ probability. But if the samples are drawn from any distribution $\bQ_{d,\ell,b,\alpha} \in \cQ_{d,k,\alpha}$, we know that $(\hat{L},\hat{B})=(\ell,b)$ with $\geq 99/100$ probability; conditioned on this event, our predictions will be correct with probability $\nicehalf + \alpha$. Thus, the count of correct predictions will reliably differentiate between the two input cases.

\mypar{Pan-privacy}: We will first prove privacy for user $i$ and intrusion time $t$. Recall that the adversary's view is the tuple $(M'_\cI(\vec{x}_{\leq t}), M'_O( M'_\cI(\vec{x})) )$; for brevity, we shall use the notation $(S_t,Z)$. If $i \leq n$ and $t \leq n$, the tuple is a post-processing of $(M_\cI(\vec{x}_{\leq t}), M_O( M_\cI(\vec{x})) )$ which we know to be $(\eps,\delta)$-private. If $i \leq n$ but $t > n$, the adversary's view is a post-processing of $ M_\cI(\vec{x})$ which is again $(\eps,\delta)$-private.

If $i > n$ but $t \leq n$, the only influence $S_t$ has on $Z$ is the choice of $(\hat{L},\hat{B})$; it suffices to prove that $Z$ is differentially private for any choice of $(\hat{L},\hat{B})$. Let $\ind(\cdot)$ be the $\zo$ indicator function. Observe that $Z \sim \Lap(1/\eps) + \sum_{u=n+1}^m \ind(\prod_{j \in \hat{L}} x_{u,j} = x_{u,d+1}\cdot \hat{B}) + \Lap(1/\eps)$. $\eps$-differential privacy follows the observation that the summation is 1-sensitive and the privacy of the Laplace mechanism.

If $i > n$ and $t > n$, we consider two further cases. When $t \geq i$, observe that $Z$ is a post-processing of $S_t$. Also observe that $S_t \sim \Lap(1/\eps) + \sum_{u=n+1}^t \ind(\prod_{j \in \hat{L}} x_{u,j} = x_{u,d+1}\cdot \hat{B})$. So we can again invoke the privacy of the Laplace mechanism. When $t<i$, we can show that $Z$ is differentially private conditioned on any realization of $S_t = (\hat{L}, \hat{B}, C_t)$: because $Z\sim \Lap(1/\eps) + C_t + \sum_{u=t+1}^m \ind(\prod_{j \in \hat{L}} x_{u,j} = x_{u,d+1}\cdot \hat{B})$ and $i \in [t+1,m]$, we invoke the privacy of the Laplace mechanism one more.

\mypar{Bound on TV distance}: Now we show that the total variation distance between $M'(\bU^m)$ and $M'(\bQ^m_{d,L,B,\alpha})$ is larger than a constant. Notice that, for any $\tau \in \R$,
\begin{align*}
&\dtv(M'(\bQ^m_{d,L,B,\alpha}) , M'(\bU^m)) \\
\geq{}& \pr{}{M'(\bQ^m_{d,L,B,\alpha}) > \tau} - \pr{}{M'(\bU^m) > \tau} \\
={}& \left( \sum_{\ell\subseteq [d], |\ell|\leq k \atop b \in \pmo} \pr{}{M'(\bQ^m_{d,\ell,b,\alpha}) > \tau} \cdot \pr{}{(L,B)=(\ell,b)} \right) - \pr{}{M'(\bU^m) > \tau} \\
={}& \left( \sum_{\ell\subseteq [d], |\ell|\leq k \atop b \in \pmo} \pr{}{M'(\bQ^m_{d,\ell,b,\alpha}) > \tau ~|~ (\hat{L},\hat{B}) = (\ell,b)} \cdot \pr{}{(\hat{L},\hat{B}) = (\ell,b)} \cdot \pr{}{(L,B)=(\ell,b)} \right) - \pr{}{M'(\bU^m) > \tau} \\
\geq{}& \left( \sum_{\ell\subseteq [d], |\ell|\leq k \atop b \in \pmo} \pr{}{M'(\bQ^m_{d,\ell,b,\alpha}) > \tau ~|~ (\hat{L},\hat{B}) = (\ell,b)} \cdot \frac{99}{100} \cdot \pr{}{(L,B)=(\ell,b)} \right) - \pr{}{M'(\bU^m) > \tau} \stepcounter{equation} \tag{\theequation} \label{eq:learn-parities-pan}
\end{align*}

\eqref{eq:learn-parities-pan} comes from the fact that $M$ learns parities. Notice that, conditioned on $(\hat{L},\hat{B}) = (\ell,b)$, Fact \ref{fact:signed-parity} implies $M'(\bQ^m_{d,\ell,b,\alpha}) $ is a sample from the convolution $\Bin(m-n,\nicehalf + \alpha) + \Lap(1/\eps) +\Lap(1/\eps)$ with probability $\geq \nicefrac{99}{100}$.

Meanwhile, note that the equality $\pr{X\sim \bU}{\prod_{j\in \ell} X_j = X_{d+1} \cdot b} = \nicehalf$ holds for any parity function $(\ell,b)$. Consequently, the output of the algorithm $M'(\bU^m) $ is a sample from the convolution $\Bin(m-n,\nicehalf) + \Lap(1/\eps) +\Lap(1/\eps)$.

Because $m-n = \Theta(1/\alpha\eps)$, we can use a Chernoff bound to argue that there is some $\tau$ where
\begin{align*}
\eqref{eq:learn-parities-pan} &\geq \left( \sum_{\ell\subseteq [d], |\ell|\leq k \atop b \in \pmo} \frac{99}{100} \cdot \frac{99}{100} \cdot \pr{}{(L,B)=(\ell,b)} \right) - \frac{1}{100} \\
&= \frac{99^2 - 100}{10000}
\end{align*}

Lemma \ref{lem:Q-sample-bound} and Theorem \ref{thm:main-lb} imply $m=\Omega\left( \sqrt{\binom{d}{\leq k}}/\alpha\eps \right)$ and, in turn, $n = \Omega\left( \sqrt{\binom{d}{\leq k}}/\alpha\eps \right)$.
\end{proof}

The next theorem adapts our proof to the robust shuffle privacy setting:
\begin{thm} \label{thm:learning-parities-shuffle}
If $\Pi$ is an $(\eps,\delta,\nicefrac{1}{3})$-robustly shuffle private protocol that learns width-$k$ signed parities with error $\alpha$ and $\delta \log \nicefrac{\binom{d}{\leq k}}{\delta} \ll \alpha^2 \eps^2 / \binom{d}{\leq k}$, then its sample complexity is $n=\Omega( \sqrt{\binom{d}{\leq k}}/\alpha\eps )$.
\end{thm}
\begin{proof}
We repeat the construction, this time building $M'$ atop $M^\Pi$ (Theorem \ref{thm:shuffle-to-pan}). To prove pan-privacy of $M'$, we follow the same steps as in the proof of Theorem \ref{thm:learning-parities-pan}; we do not replicate the text here.

Lower bounding the total variation distance between $M'(\bU^m)$ and $M'(\bQ^n_{d,L,B,\alpha})$ is also very similar though we do have to account for the change from $\Pi$ to $M^\Pi$:
\begin{align*}
&\dtv(M'(\bQ^m_{d,L,B,\alpha}) , M'(\bU^m)) \\
={}& \left( \sum_{\ell\subseteq [d], |\ell|\leq k \atop b \in \pmo} \pr{}{M'(\bQ^m_{d,\ell,b,\alpha}) > \tau ~|~ (\hat{L},\hat{B}) = (\ell,b)} \cdot \pr{}{(\hat{L},\hat{B}) = (\ell,b)} \cdot \pr{}{(L,B)=(\ell,b)} \right) - \pr{}{M'(\bU^m) > \tau} \\
\geq{}& \left( \sum_{\ell\subseteq [d], |\ell|\leq k \atop b \in \pmo} \pr{}{M'(\bQ^m_{d,\ell,b,\alpha}) > \tau ~|~ (\hat{L},\hat{B}) = (\ell,b)} \cdot \left(\frac{99}{100} - \frac{1}{6} \right) \cdot \pr{}{(L,B)=(\ell,b)} \right) - \pr{}{M'(\bU^m) > \tau} \\
\geq{}& \left( \sum_{\ell\subseteq [d], |\ell|\leq k \atop b \in \pmo} \frac{99}{100} \cdot \frac{247}{300} \cdot \pr{}{(L,B)=(\ell,b)} \right) - \frac{1}{100} \\
={}& \frac{99}{100} \cdot \frac{247}{300} - \frac{1}{100} = \frac{8051}{10000}
\end{align*}
Lemma \ref{lem:Q-sample-bound} and Theorem \ref{thm:main-lb} imply $m=\Omega\left( \sqrt{\binom{d}{\leq k}}/\alpha\eps \right)$ and, in turn, $n = \Omega\left( \sqrt{\binom{d}{\leq k}}/\alpha\eps \right)$.
\end{proof}

\section*{Acknowledgments}
We are grateful to Cl\'ement Canonne for many helpful discussions related to the proof of Lemma~\ref{lem:main-lb}.

\addcontentsline{toc}{section}{References}
\bibliographystyle{alpha}
\bibliography{refs}

\appendix
\section{Proofs for Supporting Facts for Theorem \ref{thm:main-lb}}

\label{sec:supporting-proofs}

For completeness, we prove the statements used by the proof of Theorem \ref{thm:main-lb}.
\begin{fact} [Fact~\ref{fact:tvchainrule} Restated]
	If $(A,B)$ and $(A,B')$ are joint distributions on the domain $\cA \times \cB$, then $$\dtv((A,B),(A,B')) \leq \ex{a \sim A}{\dtv(B \cond{A = a}, B' \cond{A=a})}$$
\end{fact}
\begin{proof}
	Given a set $T \subseteq \cA \times \cB$, define $T\cond{A = a} = \{ b : (a,b) \in T\}$.  Then, we have
	\begin{align*}
	&\dtv((A,B),(A,B')) \\
	={} &\sup_{T} \pr{}{(A,B) \in T} - \pr{}{(A,B') \in T} \\
	={} &\sup_{T} \ex{a \sim A}{\pr{}{ B \cond{A = a} \in T \cond{A=a}} - \pr{}{ B' \cond{A = a} \in T \cond{A=a}} } \\
	\leq{} &\sup_{T} \ex{a \sim A}{\dtv(B\cond{A=a}, B'\cond{A=a})} \\
	={} &\ex{a \sim A}{\dtv(B\cond{A=a}, B'\cond{A=a})}
	\end{align*}
	This completes the proof.
\end{proof}

\begin{fact} [Fact~\ref{fact:markovchain} Restated]
	If $(A,B,C)$ are jointly distributed random variables on $\cA \times \cB \times \cC$ and $A$ and $B$ are independent conditioned on $C$, then for every $a \in \mathrm{supp}(A)$, $$\dtv(B \cond{A = a}, B) \leq \dtv(C \cond{A=a}, C)$$
\end{fact}
\begin{proof}
	Let $T$ be an arbitrary subset of $\cB$, then we have
	\begin{align*}
	&\pr{}{B \in T \mid A = a} - \pr{}{B \in T} \\
	={} &\ex{c \sim C \cond{A=a}}{\pr{}{B \in T \mid A = a, C = c}} - \ex{c \sim C}{\pr{}{B \in T \mid C = c}} \\
	={} &\ex{c \sim C \cond{A=a}}{\pr{}{B \in T \mid C = c}} - \ex{c \sim C}{\pr{}{B \in T \mid C = c}} \tag{conditional independence} \\
	\leq{} &\sup_{f \from \cC \to [0,1]} \ex{c \sim C \cond{A=a}}{f(c)} - \ex{c \sim C}{f(c)} \\
	={} &\dtv(C \cond{A=a}, C)
	\end{align*}
	where the final inequality is because $f(c) = \pr{}{B \in T \mid C = c}$ is a function mapping $\cC \to [0,1]$.  Therefore we have
	$$
	\dtv(B\cond{A=a},B) = \sup_{T} \pr{}{B \in T \mid A = a} - \pr{}{B \in T} \leq \dtv(C \cond{A=a},C),
	$$
	as desired.
\end{proof}

\begin{lem} [Lemma~\ref{lem:approx-to-pure} Restated]
If $M \from \cX \to \cR$ is $(\eps,\delta)$-differentially private, then there is a randomizer $M'$ that is $(2\eps,0)$-differentially private such that
$$
\forall x \in \cX~~\dtv(M(x),M'(x)) \leq \delta
$$
\end{lem}
\begin{proof}
    Fix an arbitrary element $\overline{x} \in \cX$. We define $M'(\overline{x})$ to have the same distribution as $M(\overline{x})$.
    
    For any other $x\in \cX$, a lemma of Kairouz, Oh, and Viswanath~\cite{KairouzOV15}\footnote{See also Murtagh and Vadhan~\cite[Lemma 3.2]{MurtaghV16} for the precise form we use.} implies that there exists a tuple of distributions $(\tilde{M}_{0}^{x, \overline{x}},\tilde{M}_{1}^{x, \overline{x}},\tilde{M}_{\bot}^{x, \overline{x}},\tilde{M}_{\top}^{x, \overline{x}})$ where
    \begin{align*}
        &M(x) = \left(   \frac{e^{\eps} (1-\delta)}{1 + e^{\eps}}\right) \tilde{M}_{0}^{x, \overline{x}} + \left(\frac{1 - \delta}{1 + e^{\eps}}\right)\tilde{M}_{1}^{x, \overline{x}} + \delta \tilde{M}_{\bot}^{x, \overline{x}} \\
        &M(\overline{x}) = \left(   \frac{1-\delta}{1 + e^{\eps}}\right) \tilde{M}_{0}^{x, \overline{x}} + \left(\frac{e^{\eps}(1 - \delta)}{1 + e^{\eps}}\right)\tilde{M}_{1}^{x, \overline{x}} + \delta \tilde{M}_{\top}^{x, \overline{x}}
    \end{align*}
    With this context, we define $M'(x)$ to be the distribution $$M'(x) := \left(   \frac{e^{\eps} (1-\delta)}{1 + e^{\eps}}\right) \tilde{M}_{0}^{x, \overline{x}} + \left(\frac{1 - \delta}{1 + e^{\eps}}\right)\tilde{M}_{1}^{x, \overline{x}} + \delta \tilde{M}_{\top}^{x, \overline{x}}.$$

    By construction, we have
    $$
        \forall x \in \cX~~\dtv(M(x),M'(x)) \leq \delta
    $$  
    Also by construction, we have
    $$
    \forall R \subseteq \cR~~ e^{-\eps} \leq \frac{\pr{}{M'(x) \in R}}{\pr{}{M'(\overline{x}) \in R}} \leq e^{\eps}
    $$
    which implies that, for every pair $x,x' \in \cX$, we have
    $$
    \forall R \subseteq \cR~~ \frac{\pr{}{M'(x) \in R}}{\pr{}{M'(x') \in R}} \leq e^{2 \eps},
    $$
    as desired.    
\end{proof}

\section{Proofs of Other Supporting Statements}

\begin{clm}[Claim \ref{clm:P-separated} Restated]
For any $\bP\neq\bP' \in \{\bU\} \cup \cP_{d,1,\alpha}$, $\dtv(\bP,\bP') \geq \alpha$.
\end{clm}
\begin{proof}
We first compute the distance between the uniform distribution and $\bP_{d,\{j\},b,\alpha}$ (for generic $j\in[d]$ and $b\in\pmo$):
\begin{align*}
\dtv(\bU, \bP_{d,\{j\},b,\alpha}) &= \half \norm{\bU - \bP_{d,\{j\},b,\alpha}}{1} \\
    &= \half \left( \sum_{x \in \cX, x_j = b} |2^{-d} - (1+2\alpha)2^{-d}| + \sum_{x \in \cX, x_j = -b} |2^{-d} - (1-2\alpha)2^{-d}| \right) \\
    &= \half \left( \alpha \cdot 2^{-d+1} \cdot 2^{d-1} + \alpha \cdot 2^{-d+1} \cdot 2^{d-1} \right) \\
    &= \alpha
\end{align*}

For any $j,j' \in [d]$ and any $b,b'\in \pmo$, we calculate the distance $\dtv(\bP_{d,\{j\},b,\alpha},\bP_{d,\{j'\},b',\alpha})$ via case analysis. In the case where $j\neq j'$, 
\begin{align*}
\dtv(\bP_{d,\{j\},b,\alpha},\bP_{d,\{j'\},b',\alpha}) ={}& \half \norm{\bP_{d,\{j\},b,\alpha} - \bP_{d,\{j'\},b',\alpha} }{1} \\
    ={}& \half \cdot \sum_{x_j=b \atop x_{j'}=b'} |(1+2\alpha)2^{-d} -(1+2\alpha)2^{-d}| + \half \cdot \sum_{x_j\neq b \atop x_{j'}\neq b'} |(1-2\alpha)2^{-d} -(1-2\alpha)2^{-d}| \\
    &+ \half\cdot \sum_{x_j=b \atop x_{j'}\neq b'} |(1+2\alpha)2^{-d} -(1-2\alpha)2^{-d}| + \half\cdot \sum_{x_j\neq b \atop x_{j'}= b'} |(1-2\alpha)2^{-d} -(1+2\alpha)2^{-d}| \\
    ={}& \half\cdot \sum_{x_j=b \atop x_{j'}\neq b'} \alpha\cdot 2^{-d+2} + \half\cdot \sum_{x_j\neq b \atop x_{j'}= b'} \alpha\cdot 2^{-d+2} \\
    ={}& \half \left(\alpha\cdot 2^{-d+2} \cdot 2^{d-2} + \alpha\cdot 2^{-d+2} \cdot 2^{d-2} \right)\\
    ={}& \alpha
\end{align*}

In the case where $j=j'$ but $b\neq b'$, we take $b=+1$ and $b'=-1$ without loss of generality.
\begin{align*}
\dtv(\bP_{d,\{j\},+1,\alpha},\bP_{d,\{j'\},-1,\alpha}) ={}& \half \norm{\bP_{d,\{j\},+1,\alpha} - \bP_{d,\{j'\},-1,\alpha} }{1} \\
    ={}& \half \left( \sum_{x_j=+1} |(1+2\alpha)2^{-d} - (1-2\alpha)2^{-d}| + \sum_{x_j=-1} |(1-2\alpha)2^{-d} - (1+2\alpha)2^{-d}| \right) \\
    ={}& \half \left( \alpha \cdot 2^{-d+2} \cdot 2^{d-1} + \alpha \cdot 2^{-d+2} \cdot 2^{d-1} \right) \\
    ={}& 2\alpha \qedhere
\end{align*}
\end{proof}

\begin{lem}[Lemma \ref{lem:Q-sample-bound}, Restated]
For every $d \in \N$, $k \leq d$, and $\alpha \in [0,\nicehalf]$,
    $$
        \| \cQ_{d,k,\alpha} \|_{\infty \to 2}^{2} \leq \frac{4\alpha^2}{\binom{d}{\leq k}}
    $$
\end{lem}

\begin{proof}
The proof proceeds almost identically with the proof of Lemma \ref{lem:P-sample-bound}. Recall that we now take $\cX=\pmo^{d+1}$. We begin by expanding the definition of the $(\infty \to 2)$ norm:
\begin{align*}
\| \cQ_{d,k,\alpha} \|_{\infty \to 2}^{2} &= \sup_{f \from \cX \to [\pm 1]} \sum_{\bQ \in \cQ_{d,k,\alpha}} \frac{1}{|\cQ_{d,k,\alpha}|} \cdot \left( \ex{x \sim \bQ}{f(x)} - \ex{x \sim \bU}{f(x)}  \right)^2 \\
    &= \sup_{f \from \cX \to [\pm 1]} \sum_{t\subseteq [d], |t|\leq k \atop b\in \{\pm 1\}} \frac{1}{|\cQ_{d,k,\alpha}|} \cdot \left( \sum_{x \in \cX } f(x)\cdot ( \bQ_{d,t,b,\alpha}(x) - \bU(x)) \right)^2 \\
    &= \sup_{f \from \cX \to [\pm 1]} \frac{1}{2\binom{d}{\leq k} + 2} \cdot \sum_{t\subseteq [d], |t|\leq k \atop b\in \{\pm 1\}} \left( \sum_{x \in \cX } f(x)\cdot ( \bQ_{d,t,b,\alpha}(x) - \bU(x)) \right)^2 \stepcounter{equation} \tag{\theequation} \label{eq:Q-sample-bound-1}
\end{align*}
The final equality comes from Fact \ref{fact:Qdk-size}.
Note that \eqref{eq:Qtb} is equivalent to $\bQ_{d,t,b,\alpha}(x) = (1 + 2\alpha b \cdot \prod_{i \in t} x_i \cdot x_{d+1})2^{-d-1}$. We also have from Fact \ref{fact:Qdk-mixture} that $\bU(x)=2^{-d-1}$. Thus,
\begin{align*}
\eqref{eq:Q-sample-bound-1} &= \sup_{f \from \cX \to [\pm 1]} \frac{1}{2\binom{d}{\leq k} + 2} \cdot \sum_{t\subseteq [d], |t|\leq k \atop b\in \{\pm 1\}} \left(  \sum_{x \in \cX } f(x)\cdot 2\alpha b \cdot \prod_{i \in t} x_i \cdot 2^{-d-1} \right)^2  \\
    &= \sup_{f \from \cX \to [\pm 1]} \frac{2\alpha^2}{\binom{d}{\leq k}+1} \cdot \sum_{t\subseteq [d], |t|\leq k \atop b\in \{\pm 1\}} \left( \sum_{x \in \cX } f(x) \cdot \prod_{i \in t} x_i \cdot 2^{-d-1} \right)^2 \\
    &= \sup_{f \from \cX \to [\pm 1]} \frac{4\alpha^2}{\binom{d}{\leq k}+1} \cdot \sum_{t\subseteq [d], |t|\leq k} \left( \sum_{x \in \cX } f(x) \cdot \prod_{i \in t} x_i \cdot 2^{-d-1} \right)^2 \\
    &\leq \sup_{f \from \cX \to [\pm 1]} \frac{4\alpha^2}{\binom{d}{\leq k}} \cdot \sum_{t\subseteq [d]} \left( \sum_{x \in \cX } f(x) \cdot \prod_{i \in t} x_i \cdot 2^{-d-1} \right)^2 \stepcounter{equation} \tag{\theequation} \label{eq:Q-sample-bound-2}
\end{align*}
Define $\hat{f}(t) := \ex{X\sim \bU}{ f(X)\cdot \prod_{i\in t} X_i}$, the Fourier transform over the Boolean hypercube. This is precisely the term being squared above. So we have
\begin{align*}
\eqref{eq:Q-sample-bound-2} &= \frac{4\alpha^2}{\binom{d}{\leq k}} \cdot \sup_{f \from \cX \to [\pm 1]} \sum_{t\subseteq [d]} \hat{f}(t)^2 \\
    &= \frac{4\alpha^2}{\binom{d}{\leq k}} \cdot \sup_{f \from \cX \to [\pm 1]} \ex{X\sim \bU}{f(X)^2} \tag{Parseval's identity} \\
    &\leq \frac{4\alpha^2}{\binom{d}{\leq k}}
\end{align*}
This concludes the proof.
\end{proof}

\end{document}